\documentclass[11pt]{article}
\usepackage[utf8]{inputenc}

\usepackage{fullpage}
\usepackage{booktabs}
\usepackage[linesnumbered,noend,ruled,noline]{algorithm2e}
\usepackage{amsmath,amsthm,graphicx,url}
\usepackage{amssymb}
\usepackage{newpxtext}
\usepackage{newpxmath}
\usepackage{threeparttable}
\usepackage{tikz}
\usetikzlibrary{calc}
\usepackage{subcaption}
 \usepackage{mathtools}
\usepackage[normalem]{ulem}
\usepackage{soul}
\usepackage{graphicx}
\usepackage{tablefootnote}
\usepackage[numbers,sort&compress]{natbib}
\usepackage{enumitem}
\usepackage{comment}
\usepackage[dvipsnames]{xcolor}
\usepackage{bm}
\usepackage{footnote}
\usepackage{thmtools}
\usepackage{thm-restate}
\usepackage{hyperref}
\usepackage[nameinlink, capitalise]{cleveref-usedon}

\usepackage[suppress]{color-edits}
\DeclareUnicodeCharacter{0302}{}
\definecolor{darkgreen}{rgb}{0.0,0.4,0.0}
\addauthor{mb}{red}
\newcommand{\mbc}[1]{\mbcomment{#1}}
\newcommand{\mbe}[1]{\mbedit{#1}}

\addauthor{xt}{darkgreen}

\newcommand{\xte}[1]{\xtedit{#1}}

\definecolor{kwcolor}{rgb}{0.7,0.4,0.2}
\addauthor{kw}{kwcolor}

\renewcommand\citep[1]{\citealp{#1}}
\renewcommand\citet[1]{\citeauthor{#1} [\citealp{#1}]}

\newtheorem{theorem}{Theorem}[section]
\newtheorem{lemma}[theorem]{Lemma}

\newtheorem{corollary}[theorem]{Corollary}

\newtheorem{observation}[theorem]{Observation}
\newtheorem*{theorem*}{Theorem}
\newtheorem*{corollary*}{Corollary}

\theoremstyle{definition}
\newtheorem{definition}[theorem]{Definition}

\newtheorem{remark}[theorem]{Remark}

\AddToHook{env/lemma/begin}{\crefalias{theorem}{lemma}}
\AddToHook{env/conjecture/begin}{\crefalias{theorem}{conjecture}}
\AddToHook{env/corollary/begin}{\crefalias{theorem}{corollary}}
\AddToHook{env/proposition/begin}{\crefalias{theorem}{proposition}}
\AddToHook{env/claim/begin}{\crefalias{theorem}{claim}}
\AddToHook{env/observation/begin}{\crefalias{theorem}{observation}}
\AddToHook{env/definition/begin}{\crefalias{theorem}{definition}}
\AddToHook{env/example/begin}{\crefalias{theorem}{example}}
\AddToHook{env/remark/begin}{\crefalias{theorem}{remark}}
\AddToHook{env/question/begin}{\crefalias{theorem}{question}}
\AddToHook{env/condition/begin}{\crefalias{theorem}{condition}}
\AddToHook{env/assumption/begin}{\crefalias{theorem}{assumption}}

\crefname{theorem}{Theorem}{Theorems}
\crefname{lemma}{Lemma}{Lemmas}
\crefname{conjecture}{Conjecture}{Conjectures}
\crefname{corollary}{Corollary}{Corollaries}
\crefname{proposition}{Proposition}{Propositions}
\crefname{claim}{Claim}{Claims}
\crefname{observation}{Observation}{Observations}
\crefname{definition}{Definition}{Definitions}
\crefname{example}{Example}{Examples}
\crefname{remark}{Remark}{Remarks}
\crefname{question}{Question}{Questions}
\crefname{condition}{Condition}{Conditions}
\crefname{assumption}{Assumption}{Assumptions}

\DeclareMathOperator*{\argmax}{arg\,max}

\newcommand*\dif{\mathop{}\!\mathrm{d}}
\DeclareMathOperator*{\E}{\mathbb{E}}

\newcommand{\bids}{\mathbf{b}}
\newcommand{\sbids}{\mathbf{s}}

\newcommand{\alloc}{\mathbf{x}}
\newcommand{\gft}{\mathtt{GFT}}
\newcommand{\profit}{\mathtt{profit}}

\newcommand{\pb}{\mathbf{p}^b}
\newcommand{\ps}{\mathbf{p}^s}
\newcommand{\val}{\mathbf{v}}
\newcommand{\cost}{\mathbf{c}}

\numberwithin{equation}{section}

\newcommand{\sellerdis}{\mathbf{F}}
\newcommand{\indicator}{\mathbf{1}}

\newcommand{\btprice}{\ifmmode\mathrm{BT}\else\textrm{BT}\fi}
\newcommand{\doublequantile}{\textsc{\textsc{MultiQuantile}}}
\newcommand{\postquantile}{\textsc{PostQuantile}}
\newcommand{\sdma}{\ifmmode\mathrm{MA}\else\textrm{MA}\fi}
\newcommand{\property}{cap-monotone}

\SetAlFnt{\small}
\SetAlCapFnt{\small}
\SetAlCapNameFnt{\small}
\SetAlCapHSkip{0pt}
\IncMargin{-\parindent}

\definecolor{commentc}{rgb}{0.1, 0.3, 0.9}

\SetArgSty{textnormal}

\SetCommentSty{mycommfont}

\crefname{algorithm}{Auction}{Auctions}  

\definecolor{linkc}{rgb}{0.7, 0.2, 0.3}
\definecolor{citec}{rgb}{0.2, 0.3, 0.7}
\definecolor{urlc}{rgb}{0.2, 0.4, 0.3}
\hypersetup{
    colorlinks=true,
    linkcolor=linkc,
    citecolor=citec,
    urlcolor=urlc
}

\newcommand{\BGSOM}{\text{GSOM-BIC}}
\newcommand{\BGBOM}{\text{GBOM-BIC}}

\newcommand{\mainresultbicdsictext}{
Consider any Bayesian single-dimensional matching market with independent distributions that is constrained by a downward-closed family of feasible matchings. The mechanism that randomizes between GSOM and GBOM with equal probability satisfies dominant-strategy incentive-compatibility (DSIC), ex-post individual rationality (IR), and ex-ante weak budget balance (WBB), and its expected gains-from-trade is at least a $1/3.15$ fraction of the first-best expected gains-from-trade. 

Moreover, the same allocation rule can be implemented by a variant with a different payment distribution,
resulting with a mechanism that is Bayesian incentive-compatibility (BIC) and ex-post weak budget balance (WBB), while still being ex-post IR and getting the same GFT.
}

\title{Approximating Gains-from-Trade in Matching Markets}
\date{}

\author{%
\begin{tabular}{cc}
\begin{minipage}[t]{0.4\textwidth}\centering
Moshe Babaioff\thanks{Moshe Babaioff's research is supported by the Israel Science Foundation (grant No.\@ 301/24) and by a Golda Meir Fellowship.}\\
\small Hebrew University of Jerusalem\\
\small\href{mailto:moshe.babaioff@mail.huji.ac.il}{moshe.babaioff@mail.huji.ac.il}
\end{minipage}
&
\begin{minipage}[t]{0.4\textwidth}\centering
Aviad Rubinstein\thanks{Aviad Rubinstein's research is supported by the David and Lucile Packard Fellowship}\\
\small Stanford University\\
\small\href{mailto:aviad@cs.stanford.edu}{aviad@cs.stanford.edu}
\end{minipage}
\\[8ex]
\begin{minipage}[t]{0.4\textwidth}\centering
Xizhi Tan\thanks{Xizhi Tan's research is partially supported by David and Lucile Packard Fellowship}\\
\small Stanford University\\
\small\href{mailto:xizhi@stanford.edu}{xizhi@stanford.edu}
\end{minipage}
&
\begin{minipage}[t]{0.4\textwidth}\centering
Kangning Wang\\
\small Rutgers University\\
\small\href{mailto:kn.w@rutgers.edu}{kn.w@rutgers.edu}
\end{minipage}
\end{tabular}%
}

\begin{document}

\maketitle

\thispagestyle{empty}
\setcounter{page}{0}

\begin{abstract}
A central challenge in mechanism design is to develop truthful trade mechanisms that maximize the expected gains-from-trade (GFT) in two-sided markets with strategic agents. As achieving the full GFT is generally impossible, much of the literature has focused on constant-factor approximations. Existing results, however, are limited to the highly structured settings of bilateral trade and double auctions, in which every buyer can trade with every seller.

We consider the significantly more general setting of two-sided matching markets with arbitrary downward-closed constraints on the family of allowed matchings. For this setting, we present a simple randomized truthful mechanism that guarantees a constant-factor approximation to the optimal expected GFT. This result also resolves an open problem posed by Cai, Goldner, Ma, and Zhao (2021).
\end{abstract}

\setcounter{page}{1}

\section{Introduction}

Two-sided markets are a cornerstone of modern economies, facilitating trade and resource allocation in settings ranging from traditional stock exchanges and labor markets to contemporary online platforms for advertising, ride-sharing, and accommodation. 
A central goal in designing such markets is to achieve economic efficiency: ensuring that trades occur whenever they create additional value (gains-from-trade), thereby maximizing the total welfare generated. However, a fundamental obstacle arises from the presence of private information. Market participants---buyers with private valuations and sellers with private costs---act strategically to maximize their own utility. This self-interested behavior often conflicts with the market designer's goal of achieving social efficiency.

{This tension is formalized by the celebrated impossibility theorem of Myerson and Satterthwaite~\cite{MS83}: even in the simplest bilateral trade setting with a single buyer and a single seller, no mechanism can always achieve social efficiency (i.e., trade exactly when the buyer's value exceeds the seller's cost) while simultaneously satisfying three standard requirements: (i) incentive compatibility (IC), which ensures that truthful reporting of private information is optimal for the participants; (ii) individual rationality (IR), which guarantees that no agent is worse off by participating; and (iii) (weak) budget balance (BB), which requires that the mechanism never runs at a deficit. In other words, any mechanism satisfying these three conditions must, for some valuation profiles, fail to execute all socially beneficial trades.}

{Given this impossibility result, Myerson and Satterthwaite~\cite{MS83} characterized the welfare-maximizing mechanism among those satisfying IC, IR, and BB. This constrained-optimal mechanism is often referred to as the ``second-best.'' A fundamental question is the size of the resulting efficiency gap: the multiplicative loss of the second-best relative to the unconstrained ``first-best.'' While this gap was recently shown to be constant in bilateral trade~\cite{dmsw22}, it remains unknown for \emph{matching markets}, where multiple buyers and sellers interact under feasibility constraints specified by a bipartite graph so that a buyer can trade with a seller only if an edge connects them.}


{Beyond the efficiency gap, the second-best mechanism is often too complex and distribution-sensitive for practical use. Even in the simple bilateral trade setting, its allocation and payment rules depend intricately on the underlying value distributions~\cite{MS83}. This motivates the study of \emph{simple} mechanisms that are robust and easy to implement.}



{Since full efficiency is unattainable and complex mechanisms are undesirable, the central question shifts to \emph{approximation by simple mechanisms}: what is the maximum fraction of the first-best efficiency that can be guaranteed by a simple mechanism that respects the constraints of IC, IR, and BB?}

{\paragraph{Gains-from-trade in matching markets.} Market efficiency is usually measured against one of two benchmarks: \emph{social welfare} and \emph{gains-from-trade} (GFT). Social welfare is the total value of the final allocation, whereas GFT measures the net increase in welfare (i.e., the sum of $v_i - c_j$ for all trades $(i, j)$ that occur, where $v_i$ is buyer $i$'s value and $c_j$ is seller $j$'s cost). While maximizing the two objectives is mathematically equivalent, obtaining a multiplicative approximation for them is vastly different, with GFT being harder to approximate. In this work, we focus on the more demanding challenge of approximating the GFT. Specifically, we address a question that was described by~\cite{CaiGMZ21} as ``one of the major open problems in two-sided markets'':
\begin{quote}
    \emph{Is it possible to obtain a constant-factor approximation to the optimal GFT in matching markets?}
\end{quote}
}

We build on a long line of work that makes partial progress towards this holy grail. The following results are of particular relevance:
\begin{itemize}
    \item {\bf Second-best approximation:} The work of \cite{BCWZ17} characterized the mechanism that maximizes total expected sellers' profit as well as the mechanism that maximizes total expected buyers' profit, and showed that randomizing between these two mechanisms achieves a $1/2$-approximation to the \emph{second-best} benchmark---the maximum GFT achievable by any IC, IR, and BB mechanism. This left open the question of whether any IC, IR, and BB mechanism can approximate the \emph{first-best} benchmark---the optimum without respecting incentive constraints.
    \item {\bf Single edge:} The work of \cite{dmsw22} showed that in the case of a single edge (``bilateral trade''), randomizing between these two mechanisms indeed guarantees a constant fraction of the first-best GFT. Subsequent work has further refined this approach, tightening the approximation constant \cite{DBLP:conf/wine/Fei22} and providing intuitive geometric arguments \cite{hw25}. 
    Still, the question of general matching-market graphs 
    has been left open. 
    \item {\bf Complete bipartite graph:} The single-edge approximation can be extended to the case of a complete bipartite graph by carefully combining it with the \emph{trade reduction mechanism} of~\cite{MCAFEE}.
    The trade reduction mechanism guarantees at least one half of the GFT if the first-best contains at least two trades; otherwise, the problem reduces to a bilateral trade problem \cite{BabaioffCGZ18}.
\end{itemize}

All prior results leave open the question of whether it is possible to obtain a constant-factor approximation to the first-best GFT in matching markets, let alone their generalization to markets with downward-closed feasibility constraints. These constraints—where any subset of a feasible set of trades is also feasible—allow our results to encompass not only standard bipartite matchings but also more complex market structures with capacity constraints or sub-market restrictions.

\subsection{Our Results}

This paper resolves this open problem by proving that a simple and natural randomized mechanism achieves a constant-factor approximation to the first-best expected GFT in single-dimensional matching markets. Our setting is highly general, accommodating an arbitrary number of single-dimensional unit-demand buyers and unit-supply sellers with private values and costs drawn independently from arbitrary distributions, and allowing for any downward-closed feasibility constraints on the set of possible trades.
{Specifically, we analyze a mechanism that randomly chooses (with equal probability) between two core mechanisms: the  Generalized Sellers-Offering Mechanism (GSOM) and the Generalized Buyers-Offering Mechanism (GBOM), introduced by \cite{BCWZ17}. }
While prior work proved that the original mechanisms achieve a 1/2-approximation to the second-best GFT, it was not known whether this approach guarantees any constant fraction of the first-best GFT. Our core contribution is proving that this simple mechanism achieves at least a $1/3.15$ fraction of the first-best GFT, matching the best-known approximation guarantee for the special case of single-buyer single-seller bilateral trade.
\mbe{
\begin{theorem}\label{thm:main_result_BIC_DSIC_formal}
\mainresultbicdsictext
\end{theorem}

\begin{remark}
The BIC mechanism above is created by changing the payment rules of GSOM and GBOM, redistributing their payments in away that result in  ex-post WBB mechanisms. 
We refer to these variants as \BGSOM\ and \BGBOM, respectively. As GSOM (GBOM) and \BGSOM\ (\BGBOM) have the same allocation rules, we refer to them interchangeably, and only distinguish when discussing payments (and the implied incentives and budget properties). 
\end{remark}

\begin{remark}
Both the DSIC and the BIC versions of the mechanism are WBB, so some of the gains go to the mechanism (and not to the sellers and buyers), yet our proof shows that the first-best GFT approximation holds even if we only count the profits of the agents (and not those of the mechanism). Additionally, we can turn the DSIC and ex-ante WBB mechanism into an ex-ante SBB mechanism (while keeping all other properties), by distribution the mechanisms' ex-ante gains to the agents ex-ante.
\end{remark}
}

This result closes a fundamental gap in the literature on two-sided markets. It demonstrates that the strong first-best approximation guarantees recently established for the simple bilateral trade setting can, in fact, be extended to the substantially richer and more realistic setting of general matching markets. 

Furthermore, our result for the single-dimensional market implies a positive resolution for the \emph{multi-dimensional unit-demand} setting as well,\footnote{A buyer that has a multi-dimensional unit-demand valuation has a value for each item (hence ``multi-dimensional''), and her value for a set is the highest value of any item in that set (hence ``unit-demand'').}
due to a formal reduction established by \cite{CaiGMZ21}. They show that if the multiplicative gap between the first-best and second-best GFT in the single-dimensional matching market is bounded above by a constant $c$, then their mechanism achieves a $1/(2c)$-approximation of the first-best GFT in the multi-dimensional unit-demand setting. This reduction leads to the following implication.
\begin{corollary*}[Informal]
A constant-factor approximation to the first-best GFT is achievable in the market with one multi-dimensional unit-demand buyer 
and $n$ unit-supply sellers. Specifically, the mechanism proposed by Cai, Goldner, Ma, and Zhao \cite{CaiGMZ21} guarantees at least a $1/6.3$-fraction of the first-best expected GFT.
\end{corollary*}

\subsection{Technical Overview}

Our main theorem works for arbitrary downward-closed constraints on the family of allowed matchings, but for this overview, it is instructive to think of the (important) special case of matching constraints subject to an underlying bipartite graph of allowed trades. Our final mechanism is quite simple (randomizing between GSOM and GBOM), 
but the analysis is challenging. In particular, the core challenge is to relate the performance of GSOM/GBOM to the first-best GFT. This was a long-standing open problem even for the special case of a single edge (i.e., bilateral trade) until the result of~\cite{dmsw22}. In the single-edge setting, \cite{dmsw22, DBLP:conf/wine/Fei22, hw25} consider the use of random-offerer mechanisms, which can be viewed as a special case of randomizing between GSOM and GBOM. While it suffices to show that the GFT (the sum of profits of both the seller and the buyer) of this mechanism is high enough to approximate the first-best, the work of \cite{dmsw22} proves a stronger result: even just the expected profit of the (random) offerer alone suffices to achieve a constant approximation.

At a high level, our approach for analysis is to decompose our two-sided market instance into separate instances of the well-understood special case of a single edge. While this plan sounds simple, the crux of the difficulty is in coming up with the right decomposition. The first desideratum from our decomposition is: 
\begin{itemize}
    \item {\bf  Desideratum 1 (informal):} 
    The first-best GFT decomposes into the sum of the first-best GFTs of the bilateral trade instances. 
\end{itemize}

In order to explain our second---and trickier---desideratum from the decomposition, we first add some detail about the proof plan. We explain from the perspective of GSOM, and the details for GBOM are analogous. 
Conceptually, GSOM can be viewed as a mechanism run by a ``lawyer'' who acts on behalf of all sellers. The lawyer auctions the items to buyers using a mechanism that maximizes the expected total sellers' profit, which is later distributed among the sellers via a BIC mechanism. Our goal is thus to establish a lower bound on this optimal expected total sellers' profit. To this end, we introduce a hypothetical, for-analysis-only auction that the lawyer could run to extract profit from buyers. This hypothetical auction should align with the decomposition of our instance; in particular, the expected total profit from the hypothetical auction should be at least the sum of the expected profit of all sellers in all the decomposed bilateral trade instances. This requirement brings us to our second desideratum from our decomposition:
\begin{itemize}
    \item {\bf Desideratum 2: (informal)} 
    There is a feasible sellers' (resp., buyers') mechanism whose expected profit is the sum of sellers' (resp., buyers') profit in the decomposed bilateral trade instances. (The profits of these for-analysis lawyer mechanisms serve to lower bound the respective profits of GSOM/GBOM.)

\end{itemize}

Before explaining our decomposition, let us see how we can complete the analysis assuming our two desiderata: 

\begin{align*}
    \gft(\text{market}) & = \sum_{e \in \text{bilateral trade instances}} \gft(e) && \text{(Desideratum 1)} \\
    & = \sum_{e \in \text{bilateral trade instances}} O(\profit(e)) && \text{(\cite{dmsw22})} \\
    & \le O(\profit(\text{market})) && \text{(Desideratum 2)} 
\end{align*}
Here $\profit()$ in the second line refers to the expected profit of a random offerer, and in the third line to the expected profit of randomizing between GSOM and GBOM.

The remainder of this overview focuses on the introduction of the decomposition.

As discussed above, our approach is to decompose the market into a set of bilateral trades, and it is natural to do so by computing a maximum-weight matching, and then running the constant-factor bilateral trade mechanism of~\cite{dmsw22} for each edge separately. However, it is not immediately clear with respect to which graph we should compute a maximum-weight matching. 


\paragraph{Preliminaries: constant factor in a single-edge graph~\cite{dmsw22}.} For a graph with a single edge (``bilateral trade''), all known constant-factor approximations~\cite{dmsw22,DBLP:conf/wine/Fei22,hw25} follow the same blueprint: The mechanism picks one of the agents (buyer/seller) at random and lets her propose a take-it-or-leave-it offer to the other agent. It is assumed that the offerer will use a Bayesian profit-maximizing price. The analysis lower-bounds the expected GFT by lower-bounding the offerer's expected profit. The offerer's expected profit, in turn, is lower bounded by the expected profit of two specific mechanisms that the agents could run (here we call them \doublequantile~and \postquantile; details to follow in the technical sections).
Finally, the expected profit for those mechanisms is lower-bounded by a delicate accounting argument where the profit from certain good realizations is used to cover lost potential profit on other realizations which are bad.

\paragraph{Candidate graph 1: the ex-post GFT graph.} For matching markets, the first and arguably most natural candidate graph to consider is to take the underlying constraint graph and assign to each edge $(i,j)$ the weight $v_i - c_j$ (i.e., the difference between the respective realized value and cost). Note that these weights can be negative, so we can discard negative-weight edges (which does not affect the optimum) and take the maximum-weight matching. While natural, there are a few obstacles to execute the plan with this decomposition using this approach:
\begin{enumerate}
    \item The constant-factor mechanism of~\cite{dmsw22} is inherently Bayesian: the offerer uses her prior over the other agent to choose a profit-maximizing price. We clearly cannot use this mechanism with the realized value/cost!
    \item The known analyses of~\cite{dmsw22,DBLP:conf/wine/Fei22,hw25} require covering the lost profit from certain bad realizations. How should we handle the case where, in the good realizations, the buyer/seller is matched with other agents?
    \item Of course, our decomposition assumed that we know the true realized values and costs, but a mechanism would need to elicit those truthfully.
\end{enumerate}


\paragraph{Candidate graph 2: the ex-ante expected GFT graph.} In light of the aforementioned challenges with decomposition based on the ex-post realizations, it is natural to ask whether there is a good decomposition based on the ex-ante expected GFT on each edge. However, this approach fails even more fatally---it is quite possible that the expected GFT on every edge is negative! (Consider, for example, a single edge where the buyer's value is uniform from $[0,2]$ and the seller's cost is uniform from $[1,3]$.)

\paragraph{Candidate graphs 3,4,$\ldots$} There are many other variants one could consider. For example, one can take the expected GFT on each edge but discard the negative parts, take ex-post realizations of the buyers and keep the Bayesian uncertainty of the sellers, etc.

\xte{
\paragraph{Our Analysis Approach: Regrouping and the Meta-Auction.} As noted, our solution does not require designing a new, complex mechanism; rather, we analyze a simple, existing one proposed by \cite{BCWZ17}.
Because GSOM and GBOM are defined to be the optimal profit-maximizing auctions for the sellers and buyers respectively, we can lower-bound their profit by analyzing the profit of \emph{any} feasible, truthful mechanism. 

To achieve this, we introduce a hypothetical, for-analysis-only mechanism which we call the \emph{Meta-Auction}. The Meta-Auction, when paired with appropriate Bilateral-Trade allocation rule, can be used to lower bound either the GSOM or the GBOM profits. It computes the maximum-GFT matching based on ex-post realizations, but crucially, it regroups the realizations to form independent Bayesian bilateral trade instances for each edge. For a given realization, the Bayesian bilateral trade instances that are executed are the ones created for every realized edge in the maximum-GFT matching. The Bayesian instance for edge $(i,j)$ in this matching is based on ``regrouping'' all realizations where:
\begin{itemize}
    \item All buyers' values (\emph{including} $i$) are fixed, and all sellers' costs \emph{excluding} $j$ are fixed (i.e., different realizations of seller $j$'s cost $c_j$ are grouped together).
    \item Edge $(i,j)$ is in the maximum-GFT matching.
\end{itemize}

This regrouping of realizations recovers some of the Bayesian aspects of the bilateral trade problem, but quite minimally: the buyer's value is fixed, and the seller's distribution is capped to the costs where $(i,j)$ is in the maximum-GFT matching. Nevertheless, it turns out that this is sufficient to fully recover the analysis of the bilateral trade case in the framework of~\cite{dmsw22,DBLP:conf/wine/Fei22,hw25}.

Deferring a formal proof to the technical sections, we can informally gain a sense of the analysis by revisiting our earlier objections to candidate graph 1, and seeing how the Meta-Auction resolves them:
\begin{enumerate}
    \item The constant-factor mechanism of~\cite{dmsw22} is inherently Bayesian. It turns out that the subroutines we use for our Meta-Auction analysis (\doublequantile~and \postquantile) only need the Bayesian prior over the seller's costs. 
    \item The known analyses of~\cite{dmsw22,DBLP:conf/wine/Fei22,hw25} require a delicate accounting argument across different realizations in the entire distribution of the instance. We can use their analysis in a black-box manner to show that the Meta-Auction's expected profit on each edge approximates the expected GFT from that edge, and then use a global argument to cover the entire GFT of the market by simply summing the edges.  
    \item 
    While the actual mechanisms we run (GSOM/GBOM) are already known to be truthful, we must prove our hypothetical Meta-Auction is also incentive-compatible for its profit to serve as a valid lower bound. We use a monotonicity argument to show that no agent can manipulate their match in the Meta-Auction by misreporting their value/cost: they are either matched to a unique candidate agent, or not matched at all. 
    
    A related, subtle complication is that a buyer's reported value endogenously changes the ``regrouped'' distribution of costs (the cap on the seller's distribution). We identify a structural property, which we term \emph{cap-monotonicity}, that is satisfied by \doublequantile. This property ensures the Meta-Auction remains incentive-compatible despite the buyer's ability to influence the distribution they face.
\end{enumerate}}

\subsection{Related Work}

\paragraph{Gains-from-trade approximation.} 
The impossibility result of Myerson and Satterthwaite \cite{MS83} states that even for bilateral trade, no IC, IR, and BB mechanism can be fully efficient. This celebrated result, a major citation in Myerson's Nobel Prize, spurred a long line of research. One branch of this work has aimed at designing simple, practical mechanisms that approximate the optimal gains-from-trade (GFT). Early work focused on settings with additional assumptions on the distribution: McAfee \cite{mcafee2008gains} considered the fixed-price mechanism that posts the same price to the buyer and the seller (who trade if both accept their prices), and showed that the mechanism achieves a $1/2$-approximation to the first-best 
GFT under the condition that the median of the buyer's value distribution is higher than that of the seller's. Blumrosen and Mizrahi \cite{BM16} showed that a simple seller-offering mechanism achieves a $1/e$-approximation (improved to $1/(e-1)$ by \cite{DBLP:conf/wine/Fei22}) under a monotone hazard rate assumption on the buyer's distribution.

In the general setting without distributional assumptions (other than independence), the work of \cite{BCWZ17} showed that randomizing between seller-offering and buyer-offering mechanisms achieves a $1/2$-approximation of the \emph{second-best} GFT (that is, the maximum GFT achievable by any IC, IR, BB mechanism) in the bilateral trade setting. Their results generalize to matching markets, as we will discuss later. It was a long-standing open question whether a constant-factor approximation to the first-best GFT was possible for arbitrary independent distributions until it was resolved for bilateral trade by Deng, Mao, Sivan, and Wang \cite{dmsw22}. They showed that the same mechanism designed in \cite{BCWZ17} in fact guarantees a constant fraction (at least $1/8.23$) of the first-best GFT. Later on, the constant in this approximation was improved by \cite{DBLP:conf/wine/Fei22} to $1/3.15$ by refining the analysis of the same mechanism. 
The work of \cite{hw25} provided geometric proofs for a simple $1/4$-approximation and a $1/3.15$-approximation that matches the constant of \cite{DBLP:conf/wine/Fei22}.

Our main theorem leverages this state-of-the-art $1/3.15$-approximation from the bilateral trade literature as a crucial black box in our proof for general matching markets with downward-closed feasibility constraints. Our contribution can be seen as successfully ``lifting'' this powerful understanding of bilateral trade to the much more general and complex multi-agent setting. As for hardness results, in general, no mechanism can do better than an $2/e$-approximation against the first-best \cite{BM16}, and randomization between seller-offering and buyer-offering mechanisms with equal probability is \emph{not} a $1/2$-approximation against the first-best \cite{DBLP:journals/corr/abs-2111-07790,CaiGMZ21,DBLP:journals/corr/abs-2603-08679}. Recently, gains-from-trade approximations have been studied in settings with a broker \cite{DBLP:conf/soda/HajiaghayiHPS25} and with a sampling component \cite{DBLP:conf/sigecom/DengMS0W25}.

\paragraph{Welfare approximation in bilateral trade.}
Other than the gains-from-trade objective, there has been much work focusing on approximating the related benchmark of first-best social welfare; see, e.g., \cite{BlumrosenD21,KangPV22,CaiW23,DBLP:conf/stoc/LiuR023,DBLP:conf/stoc/DobzinskiS24,DBLP:conf/sigecom/DobzinskiEGST25}. Mathematically, the expected social welfare is always equal to the expected gains-from-trade plus the seller's expected (non-negative) cost, and so any $\alpha$-approximation to the first-best gains-from-trade is at least an $\alpha$-approximation to the first-best welfare, but the converse is not true.

\paragraph{Two-sided markets.}
Beyond single-buyer single-seller bilateral trade, there has been extensive research in more general settings such as double auctions and two-sided markets with multiple buyers and sellers. In an early work, McAfee \cite{MCAFEE} considered the trade reduction mechanism in the double auction setting. 
He showed that by sometimes giving up one trade with the least gains, one can construct an IC, IR, and BB mechanism with high GFT.
Later work discovered new insights in double auction and two-sided market settings, often from the perspective of showing approximation guarantees \cite{DuttingRT14, Colini-Baldeschi16, BCWZ17, BabaioffCGZ18, colini2020approximately, CaiGMZ21}. More recently, there has been work on understanding how competition can help improve the efficiency of simple mechanisms such as the trade reduction mechanism \cite{BabaioffGG20,cai2024power}.

Another recent line of work has focused on bilateral trade from the perspective of online learning and regret minimization \cite{cesa2021regret, DBLP:conf/nips/AzarFF22, cesa2023repeated, DBLP:conf/atal/BolicCC24,bernasconi2024no, cesa2024bilateral,DBLP:conf/nips/BachocCCC24,DBLP:journals/corr/abs-2504-04349,DBLP:journals/corr/abs-2601-16412}. In these settings, a learning algorithm must make decisions online repeatedly when the values and costs are either adversarially chosen or are drawn from unknown distributions. These algorithms usually aim at minimizing the cumulative regret over a long period of time.

\section{Preliminaries}\label{sec:prelim}

\paragraph{Matching market.}
We consider a general two-sided matching market with $m$ unit-demand buyers and $n$ unit-supply sellers. Each buyer $i$ has a private value $v_i \geq 0$ 
for receiving an item, while each seller $j$ has a private cost $c_j \geq 0$ for providing one. The value $v_i$ of each buyer $i$ is drawn from a distribution $D_i$, and the cost $c_j$ of each seller $j$ is drawn from a distribution $F_j$. We assume that for every cost $c\geq 0$ there exists a price $p$ that maximizes the expected utility $(p-c)\cdot \mathbb{P}[v\geq p]$ (and vice versa for a buyer).
All $D_i$'s and $F_j$'s are publicly known distributions that are mutually independent. We also use the notations $D_i$ and $F_j$ to denote the corresponding cumulative distribution functions (CDFs). We use $\mathbf{v} = (v_1, \dots, v_m)$ and $\mathbf{c} = (c_1, \dots, c_n)$ to denote the vectors of realized values and costs, and $\mathbf{D}=D_1\times D_2\times \ldots\times D_m$ and $\mathbf{F}=F_1\times F_2\times \ldots\times F_n$ for their respective (joint) product distributions. 
The setting is single-dimensional for both the buyers and the sellers.



Let $E \subseteq [m] \times [n] = \{(i,j) \mid i \in [m], j \in [n]\}$ be the set of all possible trading pairs between the $m$ buyers and $n$ sellers (a bipartite graph). 
This bipartite graph defines the family of feasible \emph{matchings} in the market. While standard matching markets treat any matching in a given bipartite graph as feasible, we extend this model by allowing the family of feasible matchings to be restricted to an arbitrary downward-closed family.
Formally, let $\mathcal{F} \subseteq 2^E$, with $\{\varnothing\} \in \mathcal{F}$, denote the non-empty family of feasible matchings, and assume that it satisfies the following conditions:
\begin{itemize}
    \item The family $\mathcal{F}$ is matching-based: every $M \in \mathcal{F}$ is a matching in the bipartite graph $E$;
    \item The family $\mathcal{F}$ is downward-closed: if $M \in \mathcal{F}$ and $M' \subseteq M$, then $M' \in \mathcal{F}$.
\end{itemize}
This formulation allows for general constraints that may disallow certain combinations of trades, provided that the family of feasible matchings remains downward-closed. As a simple example, our formulation can accommodate the \emph{additional} constraint that at most $k$ trades can take place in the entire matching market.

An allocation is determined by a matching $M \in \mathcal{F}$, where each matched buyer $i$ receives a unit from their matched seller $j$, and any unmatched seller retains their unit. We denote an allocation using indicator variables $\mathbf{x} \in \{0,1\}^{E}$, where $x_{ij} = 1$ if and only if $(i,j) \in M$. We note that our setting captures the well-studied settings of bilateral trade (where there is one buyer and one seller; $m = n = 1$), double auctions (where any buyer can trade with any seller; $\mathcal{F}$ is the family of all matchings in the complete bipartite graph), and more narrowly defined matching markets where $\mathcal{F}$ is the collection of all matchings of a given bipartite graph.

Our goal is to maximize the \emph{gains-from-trade} (GFT). For a given allocation (matching) $M \in \mathcal{F}$ and a realization $(\mathbf{v}, \mathbf{c})$, the gains-from-trade are $\gft(M, \mathbf{v}, \mathbf{c}) = \sum_{(i,j) \in M} (v_i - c_j)$.
For an instance with realization $(\mathbf{v}, \mathbf{c})$ and with the family of feasible matchings being $\mathcal{F}$, the expected optimal GFT (referred to as the \emph{first-best}) is 
\[
\gft^{*} = \E_{\mathbf{v}, \mathbf{c}}\left[\max_{M \in \mathcal{F}}\gft(M, \mathbf{v}, \mathbf{c})\right].
\]
\paragraph{Mechanism notation.}
We consider the mechanism design problem where the buyers' values and the sellers' costs are private information. A (direct revelation) mechanism takes as input the public distributions $\mathbf{D}$ and $\sellerdis$, buyers' bids $\bids = (b_1, \dots, b_m)$ and sellers' bids $\sbids = (s_1, \dots, s_n)$, and gives as output the allocation $\mathbf{x}(\bids, \sbids) \in [0,1]^{m \times n}$ as well as the payments $\pb(\bids, \sbids) \in \mathbb{R}_{\geq 0}^{m}$ and $\ps(\bids, \sbids) \in \mathbb{R}_{\geq 0}^{n}$ for buyers and sellers, respectively. Here, $x_{ij} = \mathbf{x}_{ij}(\bids, \sbids) \in [0, 1]$ denotes the fraction of the item of seller $j$ that is transferred to buyer $i$, 
and $p^b_i = p^b_i(\bids, \sbids) \geq 0$ and $p^s_j = p^s_j(\bids, \sbids)\geq 0$ denote the payment charged to buyer $i$ and the payment paid to seller $j$, respectively. The agents are assumed to be risk-neutral with quasi-linear utilities. The utility for a buyer $i$ with true value $v_i$ is $v_i \cdot \sum_{j} x_{ij} - p_i^b$, and the utility for a seller $j$ with true cost $c_j$ is $p_j^s - c_j \cdot \sum_{i} x_{ij}$.

A mechanism is \emph{dominant-strategy incentive-compatible} (DSIC) if each agent maximizes her utility when she reports truthfully, no matter what the other agents do. Formally, a  mechanism is DSIC for the sellers if for any seller $j$, any $c_j$, any deviation $s_j$, and any reports $\mathbf{b}$ and $\sbids_{-j}$,
we have
\begin{align}\label{eq:seller_dsic}
p^s_j(\mathbf{b},(c_j,\sbids_{-j})) - c_j \cdot\sum_{i}x_{ij}(\mathbf{b},(c_j,\sbids_{-j}))\geq p^s_j(\mathbf{b},(s_j,\sbids_{-j})) - c_j \cdot\sum_{i}x_{ij}(\mathbf{b},(s_j,\sbids_{-j})). \tag{DSIC-S}
\end{align}

Analogously, a mechanism is DSIC for the buyers if for any buyer $i$, any $v_i$, any deviation $b_i$, and any reports $\mathbf{b}_{-i}$ and $\sbids$,
we have
\begin{align}\label{eq:buyer_dsic}
v_i \cdot\sum_{j}x_{ij}((v_i, \mathbf{b}_{-i}),\sbids)-p^b_i((v_i, \mathbf{b}_{-i}),\sbids) \geq v_i \cdot\sum_{j}x_{ij}((b_i, \mathbf{b}_{-i}),\sbids)-p^b_i((b_i, \mathbf{b}_{-i}),\sbids). \tag{DSIC-B}
\end{align}
A mechanism is DSIC if it is DSIC for both the sellers and the buyers.

A mechanism is \emph{Bayesian incentive-compatible} (BIC) if each agent maximizes her expected utility by reporting truthfully, given that all others do. 
Formally, a mechanism is BIC for the sellers if for any seller $j$, any $c_j$, and any deviation $s_j$, we have
\begin{align}\label{eq:seller_bic}
\E_{\mathbf{v},\mathbf{c}_{-j}}\left[p^s_j(\mathbf{v},(c_j,\mathbf{c}_{-j})) - c_j \cdot\sum_{i}x_{ij}(\mathbf{v},(c_j,\mathbf{c}_{-j}))\right]\geq \E_{\mathbf{v},\mathbf{c}_{-j}}\left[p^s_j(\mathbf{v},(s_j,\mathbf{c}_{-j})) - c_j \cdot\sum_{i}x_{ij}(\mathbf{v},(s_j,\mathbf{c}_{-j}))\right].\tag{BIC-S}
\end{align}
Analogously, a mechanism is BIC for the buyers if for any buyer $i$, any $v_i$, and any deviation $b_i$, we have
\begin{align}\label{eq:buyer_bic}
\E_{\mathbf{v}_{-i},\mathbf{c}}\left[v_i \cdot \sum_{j}x_{ij}((v_i, \mathbf{v}_{-i}),\mathbf{c})-p^b_i((v_i, \mathbf{v}_{-i}),\mathbf{c})\right]\geq \E_{\mathbf{v}_{-i},\mathbf{c}}\left[v_i \cdot \sum_{j}x_{ij}((b_i, \mathbf{v}_{-i}),\mathbf{c})-p^b_i((b_i, \mathbf{v}_{-i}),\mathbf{c})\right]. \tag{BIC-B}
\end{align}
A mechanism is BIC if it is BIC for both the sellers and the buyers.

A mechanism satisfies ex-post \emph{individual rationality} (IR) if each agent receives non-negative utility from participating and reporting truthfully, no matter what the private values and costs are. Formally, for any seller $j$, any $c_j$, and any reports $\mathbf{b}$ and $\sbids_{-j}$, we have
\begin{align}\label{eq:seller_ir}
p^s_j(\mathbf{b},(c_j,\sbids_{-j})) - c_j \cdot\sum_{i}x_{ij}(\mathbf{b},(c_j,\sbids_{-j}))\geq 0, \tag{IR-S}
\end{align}
and for any buyer $i$, any $v_i$, and any reports $\mathbf{b}_{-i}$ and $\sbids$, we have
\begin{align}\label{eq:buyer_ir}
v_i \cdot\sum_{j}x_{ij}((v_i, \mathbf{b}_{-i}),\sbids)-p^b_i((v_i, \mathbf{b}_{-i}),\sbids) \geq 0. \tag{IR-B}
\end{align}
A mechanism is ex-post \emph{individual rationality} (IR) if it is ex-post IR for both the sellers and the  buyers.


A mechanism is ex-ante \emph{weakly budget-balanced} (WBB) if, assuming truthful reporting, the expected total payment from the buyers is at least the expected total payment to the sellers:
\begin{align}\label{eq:wbb}
\E_{\mathbf{v}, \mathbf{c}}\left[\sum_{i}p^b_i(\mathbf{v},\mathbf{c})\right] \geq \E_{\mathbf{v}, \mathbf{c}}\left[\sum_{j} p^s_j(\mathbf{v},\mathbf{c})\right]. \tag{WBB}
\end{align}
A mechanism is ex-post strongly \emph{budget-balanced} (SBB) if, assuming truthful reporting, the total payment from buyers is equal to the expected total payment to sellers in every realization, i.e., $\sum_{i}p^b_i(\mathbf{v},\mathbf{c}) = \sum_{j} p^s_j(\mathbf{v},\mathbf{c})$.

The \emph{second-best} GFT is the maximum expected GFT achievable by any mechanism that satisfies all our desired constraints: BIC (for buyers and sellers), interim IR, and ex-ante WBB.

\paragraph{Implementability and Myerson's lemma.}
The matching market we have described is a \emph{single-parameter environment}: each agent's private information (their ``type'') is fully captured by a single real number, representing the value $v_i$ for a buyer or the cost $c_j$ for a seller. Furthermore, the agent utilities are quasi-linear and depend on this parameter only through its interaction with the allocation $x_{ij}$ (e.g., $v_i \cdot \sum_j x_{ij}$ or $-c_j \cdot \sum_i x_{ij}$, as seen in \cref{eq:seller_dsic,eq:buyer_dsic}). We state below the celebrated Myerson's lemma on implementability in this context.



\begin{theorem}[Implementability and Payment Formula \cite{Myerson81}]\label{thm:myerson}
An allocation rule $\alloc$ is \emph{implementable} in dominant strategies 
if and only if it is monotone. Formally:

\textbf{For single-parameter buyers in trade environments:}
\begin{enumerate}
    {\item An allocation rule $\alloc$ is implementable in dominant strategies if and only if, for every fixed $(\bids_{-i}, \sbids)$, the total allocation to buyer $i$, $\sum_j x_{ij}((b_i, \bids_{-i}), \sbids)$, is monotone non-decreasing in $b_i$.}
    \item Given $\alloc$ is monotone non-decreasing, there is a unique payment rule $\mathbf{p}$ that makes the mechanism DSIC and satisfies ex-post IR (assuming zero payment for a zero bid). It is given by the integral formula
    \[
    p_i((b_i, \bids_{-i}), \sbids) = b_i \cdot \sum_{j} x_{ij}((b_i, \bids_{-i}), \sbids) - \int^{b_i}_{0} \sum_j x_{ij}((t, \bids_{-i}), \sbids) \dif t.
    \]
\end{enumerate}
\textbf{For single-parameter sellers in trade environments:}
\begin{enumerate}
    {\item An allocation rule $\alloc$ is implementable in dominant strategies if and only if, for every fixed $(\bids, \sbids_{-j})$, the total allocation to seller $j$, $\sum_i x_{ij}(\bids, (s_j, \sbids_{-j}))$, is monotone non-increasing in $s_j$.}
    \item Given $\alloc$ is monotone non-increasing, there is a unique payment rule $\mathbf{p}$ that makes the mechanism DSIC and satisfies ex-post IR (assuming a seller not allocated receives no payment). It is given by the integral formula
    \[
    p_j(\bids, (s_j, \sbids_{-j})) = s_j \cdot \sum_{i} x_{ij}(\bids, (s_j, \sbids_{-j})) + \int^{\infty}_{s_j} \sum_{i} x_{ij}(\bids, (t, \sbids_{-j})) \dif t.
    \]
\end{enumerate}
\end{theorem}
\begin{definition}[Threshold Payments]\label{def:threshold_payment}
For a deterministic monotone allocation rule, the integral payment formulas from \cref{thm:myerson} reduce to simple \emph{threshold payments}.
\begin{itemize}
    \item For a buyer $i$, the allocation $\sum_j x_{ij}$ is a step function that is 0 for bids $b_i < b_i^*$ and 1 for $b_i \ge b_i^*$. The payment $p_i$ from the integral formula simplifies to $p_i = b_i^*$, where $b_i^* = \inf \{ b' \mid \sum_j x_{ij}(b', \cdot) = 1 \}$ is the infimum winning bid.
    
    \item Symmetrically, for a seller $j$, the allocation $\sum_i x_{ij}$ is a step function that is 1 for bids $s_j \le s_j^*$ and 0 for $s_j > s_j^*$. The payment $p_j$ simplifies to $p_j = s_j^*$, where $s_j^* = \sup \{ s' \mid \sum_i x_{ij}(s', \cdot) = 1 \}$ is the supremum winning cost.
\end{itemize}
\end{definition}

\paragraph{The profit-maximizing auctions.}
The two mechanisms at the core of our result are the \emph{Generalized Sellers-Offering Mechanism} (GSOM) and the \emph{Generalized Buyers-Offering Mechanism }(GBOM). 
\mbe{We first define their allocation rules by the optimization they conduct}: 

\begin{itemize}
\item \textbf{GSOM (Generalized Sellers-Offering Mechanism) \mbe{allocation rule}:} GSOM takes all sellers' reported costs $\sbids$ and buyers' distributions $\mathbf{D}$, and implements the allocation rule $\mathbf{x}$ that maximizes the sellers' total expected (with respect to buyers' distribution) profit among all possible BIC-B and IR-B allocation rules. 


\item \textbf{GBOM (Generalized Buyers-Offering Mechanism) \mbe{allocation rule}:} GBOM takes all buyers' reported value $\bids$ and sellers' distributions $\mathbf{F}$, and implements the allocation rule $\mathbf{x}$ that maximizes the buyers' total expected (with respect to sellers' distributions) profit among all possible BIC-S and IR-S allocation rules.
\end{itemize}

For both mechanisms, the payments $\mathbf{p}^b$ and $\mathbf{p}^s$ are the unique threshold payments (\cref{def:threshold_payment}) that make these respective allocation rules DSIC. 
\mbe{These two mechanisms were introduced and analyzed by \cite{BCWZ17}, which } 
\mbdelete{
These mechanisms were introduced and analyzed by \cite{BCWZ17} in 
different forms, but we show that the two definitions are equivalent. 
We defer the equivalence proof to \cref{app:gsom_equivalence}. \mbc{a sentence should never starts with a citation.} \cite{BCWZ17}} proved the following property of them.
\begin{theorem}[\cite{BCWZ17}]
\label{thm:feasibility}
    Both GSOM and GBOM are DSIC-B, DSIC-S, IR-B, IR-S, and ex-ante WBB.
\end{theorem}

In addition, we provide BIC versions of GSOM and GBOM, termed \BGSOM\ (\cref{def:bgsom}) and \BGBOM\ (\cref{def:bgbom}), respectively, by implementing the same allocation rule with a different per-realization payment that are ex-ante equivalent to the GSOM and GBOM. 
We do so to get the budget to always be budget-balanced (ex-post). Specifically, they are BIC, ex-post IR and ex-post WBB. The formal definition and proofs are deferred to \cref{app:gsom_equivalence}.

\begin{theorem}
    \BGSOM\ is BIC-S, DSIC-B, IR-S, IR-B and ex-post WBB; \BGBOM\ is BIC-B, DSIC-S, IR-S, IR-B and ex-post WBB.
\end{theorem}

\section{Profit Approximating the First-Best GFT}
In this section, we outline the technical roadmap to prove our main result. We begin by formally stating our main theorems. \cref{thm:main_result_BIC_DSIC_formal} is our main ``public-facing'' result, 
while \cref{thm:main_profit_technical} is the core technical result we will prove, which implies the former.

\begingroup
\renewcommand{\thetheorem}{\ref{thm:main_result_BIC_DSIC_formal}}
\begin{theorem}
\mainresultbicdsictext
\end{theorem}
\addtocounter{theorem}{-1}
\endgroup

We note that GSOM/GBOM mechanisms might not be unique; yet we refer to them as such for simplicity, as the aforementioned approximation guarantee in \cref{thm:main_result_BIC_DSIC_formal} holds for \emph{any} valid GSOM/GBOM. For instance, for equal-revenue distributions, posting any price strictly greater than $1$ constitutes a valid GSOM mechanism.  


\mbe{We defer the formal definitions and proofs of the BIC version to \cref{app:gsom_equivalence} and focus solely on the DSIC version for the rest of the paper.} 

\cref{thm:main_result_BIC_DSIC_formal} establishes the existence of a DSIC, {ex-post} IR, and ex-ante WBB mechanism that approximates the first-best. This result has a direct implication for the gap between the second-best (i.e., the optimal expected GFT achievable by any mechanism that satisfies BIC, {interim} IR, and {ex-ante W}BB) and the first-best (i.e., the optimal expected GFT, with no incentive constraints) benchmarks.

\begin{corollary}
Consider any Bayesian single-dimensional matching market with independent distributions that is constrained by a downward-closed family of feasible matchings. The second-best expected GFT (the optimal GFT achievable by any BIC, interim IR, and ex-ante WBB mechanism is at least a $1/3.15$ fraction of the first-best expected GFT.
\end{corollary}

Our main theorem also implies a constant-approximation in the setting with multi-dimensional unit-demand buyers (with different values for different items). In particular, applying this constant to the reduction in \cite{CaiGMZ21} yields the following corollary (see section \cref{sec:multi-dim} for details).

\begin{corollary}
The mechanism proposed in \cite{CaiGMZ21} (the better of ``Generalized Buyer Offering Mechanism'' and ``Seller Adjusted Posted Price'') achieves at least a $1/6.3$-approximation of the first-best GFT in the matching market setting with one multi-dimensional unit-demand buyer and $n$ unit-supply sellers.
\end{corollary}

The proof of \cref{thm:main_result_BIC_DSIC_formal} relies on the fact that the expected GFT of any ex-ante WBB mechanism is lower-bounded by its total expected profit. We prove our result by establishing the following bound on the total expected profit of our randomized mechanism.

\begin{theorem}[GFT Approximation via Profit]\label{thm:main_profit_technical}
Consider any Bayesian single-dimensional matching market with independent distributions that is constrained by a downward-closed family of feasible matchings. Let $\Pi_S(\text{GSOM})$ and $\Pi_B(\text{GBOM})$ be the optimal total expected sellers' profit and the optimal total expected buyers' profit, respectively. The total expected profit from randomizing between these two mechanisms is a $1/3.15$-approximation to the expected optimal GFT:
\[
\frac{1}{2} \Pi_S(\text{GSOM}) + \frac{1}{2} \Pi_B(\text{GBOM}) \ge \frac{1}{3.15} \gft^*.
\]
\end{theorem}

\paragraph{Proof roadmap.}
Analogous to the GFT analysis in the bilateral trade setting, our proof of \cref{thm:main_profit_technical} does not analyze the profit of GSOM and GBOM directly. Instead, we establish this bound through a multi-step profit-based reduction. The rest of the paper is structured as follows:
\begin{itemize}
    \item In \cref{sec:meta_auction}, we introduce a general mechanism template called the \emph{Meta-Auction} (\sdma). These alternative mechanisms allow the buyers/sellers profit maximizer to ``reduce'' the complex matching market problem into a set of simpler, modified bilateral trade problems. 
    We then prove the incentive and profit guarantees of the \sdma~auctions. Since GSOM and GBOM are the optimal profit-maximizing mechanisms by definition, the expected profit of the \sdma~provides a robust lower bound on the benchmark profit. Thus:
    \[
    \frac{1}{2}(\Pi_S(\text{GSOM}) + \Pi_B(\text{GBOM})) \ge \frac{1}{2}(\Pi_S(\sdma) + \Pi_B(\sdma)).
    \]

    \item We then prove (in \cref{lem:buyerprofit,lem:sellerprofit}) that for each pair of buyer $i$ and seller $j$ such that $(i,j) \in E$, the \emph{realized} profit of this Meta-Auction run by either the buyer side or the seller side is, in turn, lower-bounded by the \emph{realized} profit from running the corresponding bilateral trade auction on modified, standalone bilateral trade instances. Specifically, let $\Pi_{\btprice}(i,j)$ denote the profit from the modified bilateral trade auction on edge $(i,j)$. 
    

    \item In \cref{sec:gft_decomposition}, we prove the GFT decomposition result (\cref{lem:decomposition}), which shows that the total first-best GFT ($\gft^*$) is exactly equal to the sum of the expected GFTs from these same modified bilateral trade instances:
    \[
    \gft^* = \sum_{(i,j) \in E} \gft_{\btprice}(i,j).
    \]
\end{itemize}
By combining these steps, we reduce the complex market-wide problem to the well-studied bilateral trade setting. We then apply a known $1/3.15$-approximation from this literature (\cref{thm:bilateral_approx}) to the sum of bilateral profits, which, by our decomposition lemma, relates directly back to the total $\gft^*$ to complete the proof.

\section{Alternative Single-Side Profit-Maximizing Auctions}\label{sec:meta_auction}

We propose a general ``Meta-Auction'' framework that can be parameterized to generate auctions for either the seller side or the buyer side. It serves, for analysis purposes, as a hypothetical alternative to the auction that maximizes the total expected profit for the corresponding side. We note that this section analyzes a one-sided strategic problem. For example, when maximizing the total expected profit for the sellers, we treat the buyers as the only strategic agents from whom profit is extracted. The sellers, in this specific context, are non-strategic. 

The template is flexible: it can be run by either the buyer side or the seller side, by plugging in different Bilateral Trade Allocation Rule (\btprice) and the corresponding payment rule. For clarity, we now describe the high-level idea of the auction from the buyer-run perspective, so the auction should be DSIC-S. 
All auctions generated from this template share a common two-phase structure:
\begin{enumerate}
    \item \textbf{Global matching phase:} First, the auction identifies the set of welfare-maximizing (equivalently, GFT-maximizing) trades by computing the maximum-weight matching, $M^* \in \mathcal{F}$, based on all reports. We use a standard lexicographical tie-breaking rule (e.g., based on a fixed global ordering of all feasible matchings) and hence guarantee that $M^*$ is unique and, 
    as we show in \cref{lem:mwm_monotonicity}, the resulting allocation rule that maps reports to the selected matching preserves the necessary monotonicity properties for incentive compatibility.
    
    \item \textbf{Bilateral local pricing phase:} Second, the auction iterates through each ``approved'' pair $(i,j) \in M^*$. For each pair, it first modifies the seller's distribution as we describe in \cref{def:modified_distribution} below (creating the modified seller distribution $\bar{F}_{ij}$). Crucially, in both buyer-run and seller-run Meta-Auctions, the modification is applied solely to the seller's distribution $\bar{F}_{ij}$. As this is buyer-run, the auction then applies a deterministic bilateral trade auction $BT^B$ to determine the final allocation $x_{ij}$ for that pair. The mechanism's final payment for seller $j$ is then set as the unique threshold payment (\cref{def:threshold_payment}) that makes this entire two-phase allocation rule incentive-compatible for each seller $j$.
\end{enumerate}

A key concept that links these two phases is the modification of the seller's distribution. Given reported $(\bids, \sbids)$, for each edge $(i, j)$ belonging to the MWM $M^*$, the global matching problem is used to calculate a critical cost $\bar{c}_{ij}$. This value represents the maximum cost seller $j$ can report such that the pair $(i, j)$ remains in the maximum weight matching, holding all other bids constant:$$\bar{c}_{ij} = \sup \{ c \mid (i,j) \in M^*(\mathbf{b}, (c, \mathbf{s}_{-j})) \}.$$ This critical cost defines a modified, market-aware bilateral trade instance for the specific pair $(i,j)$ via a modified distribution $\bar{F}_{ij}$. The bilateral trade auction is then applied to this modified instance. 

\begin{definition}[Conditional Critical Cost]\label{def:criticalcost}
For any pair $(i,j) \in E$ and any fixed realization of all values and all other costs $(\mathbf{v}, \mathbf{c}_{-j})$, the \emph{conditional critical cost} $\bar{c}_{ij}(\mathbf{v}, \mathbf{c}_{-j})$ is the supremum of all costs of seller $j$ for which $(i, j)$ remains in the first-best matching $M^*$:
\[
\bar{c}_{ij}(\mathbf{v}, \mathbf{c}_{-j}) \equiv \sup \{ c \mid (i,j) \in M^*(\mathbf{v}, (c, \mathbf{c}_{-j})) \}.
\]
If $(i,j) \notin M^*(\val, (c, \cost_{-j}))$ for every $c$, we let $\bar{c}_{ij}(\mathbf{v}, \mathbf{c}_{-j}) \equiv  -\infty$.
\end{definition}

\begin{definition}[Conditional Modified Distribution]\label{def:modified_distribution}
Given the conditional critical cost $\bar{c}_{ij}(\mathbf{v}, \mathbf{c}_{-j})$, we define the \emph{conditional modified distribution} $\bar{F}_{ij}(\cdot \mid \mathbf{v}, \mathbf{c}_{-j})$ as the function
\[
\bar{F}_{ij}(x \mid \mathbf{v}, \mathbf{c}_{-j}) \equiv F_j(\min(x, \bar{c}_{ij}(\mathbf{v}, \mathbf{c}_{-j}))).
\]
This is the CDF of a modified cost, which is drawn from $F_j$ but is effectively ``censored'' at $\bar{c}_{ij}$. Any true cost $c_j$ with $c_j > \bar{c}_{ij}$ is treated as $\infty$.
\end{definition}

The formal description of this framework is presented in \cref{alg:auctiontemplate}. The goal of this Meta-Auction (\sdma) is to reduce the complex auction design in a matching market to a set of simpler, modified bilateral trade problems, one for each pair in the maximum-weight matching. We then use the \sdma\ from the buyers side as a feasible mechanism (a ``bridge'') to connect the optimal buyers profits from GBOM 
to the profits achievable in these simpler bilateral trades. 

We will show two key properties of this framework.
First, the template can be parameterized with different subroutines, and it can be made incentive-compatible for either side. For the buyer-run version, if the \btprice\ is DSIC-S, then the entire \sdma\ is DSIC-S. 
For the seller-run version, if the \btprice\ is DSIC-B, then the \sdma\ is DSIC-B, provided that the rule satisfies a mild \emph{\property} (\cref{def:btproperty}) assumption. 
Because these auctions are incentive-compatible for the strategic side, the profits they generate provide a valid lower bound for GBOM and GSOM, respectively.

The second case is particularly subtle: in the seller-run version, a buyer's report has a dual role: it serves as his bid in the local allocation pricing phase and also endogenously determines the bilateral trade instance they face (by influencing the MWM and the resulting critical cost $\bar{c}_{ij}$). The \property\ assumption ensures that buyers remain truthful despite this influence on the corresponding  ``market-aware'' sellers' modified distribution.

Second, we connect the \sdma\ profit to the profits of the modified bilateral trade instances, showing that the total \sdma\ profit is lower-bounded by the sum of the profits from these individual instances. This connection is the key that allows us to utilize existing profit--GFT guarantees from the bilateral trade literature to prove our main approximation result for the entire market.
\renewcommand*\footnoterule{}

\begin{algorithm}[ht]
\SetAlgoLined
\caption{The Meta-Auction Framework: $\sdma(T)$}
\label{alg:auctiontemplate}
\KwIn{Buyers' bids $\mathbf{b}$, Sellers' bids $\mathbf{s}$, Sellers' product distribution $\mathbf{F}$, Family $\mathcal{F}$ of feasible matchings, Side $T \in \{B, S\}$}
\KwOut{Allocation $\mathbf{x}$, Payments $\mathbf{p}$}
\BlankLine

\tcp{Phase 1: Global ex-post welfare maximization}
$M^* \leftarrow \arg\max_{M \in \mathcal{F}} \sum_{(i,j) \in M} (b_i - s_j )$; \tcp*{Unique via fixed tie-breaking}

\For{each pair $(i,j) \in M^*$}{
    \tcp{Phase 2: Local modification (always on sellers' side)}
    $\bar{c}_{ij} \leftarrow \sup \{ c \mid (i,j) \in M^*(\mathbf{b}, (c, \mathbf{s}_{-j})) \}$; \\
    $\bar{F}_{ij}(x) \leftarrow F_j(\min(x, \bar{c}_{ij}))$; \\
    \BlankLine

    \eIf{$T = B$ (Buyers-side-run)}{
        \tcp{Use subroutine that is DSIC for the sellers}
        $x_{ij} \leftarrow \btprice^{B}(\bar{F}_{ij}, b_i, s_j)$; \\
        $p^s_j \leftarrow \text{Threshold payment to seller } j$; \\
    }{
        \tcp{Use subroutine that is DSIC for the buyers}
        $x_{ij} \leftarrow \btprice^{S}(\bar{F}_{ij}, b_i, s_j)$; \\
        $p^b_i \leftarrow \text{Threshold payment from buyer } i$; \\
    }
}
\Return{$(\mathbf{x}, \mathbf{p})$};
\end{algorithm}

\paragraph{The bilateral trade allocation rule (\btprice).} 
This one-sided distribution dependence is also used in all previous constant-approximation bilateral trade analysis.
Below we list the main family of \btprice\ auctions considered in the bilateral trade analysis:
\begin{itemize}
\item Seller-Proposing (\doublequantile): Given the seller's bid $c$ and the seller distribution $F$, let $\lambda \in (0,1)$ be a constant parameter. The mechanism posts a price to the buyer of
\[
p^b =F^{-1}\left(\min\left\{1, \frac{F(c)}{\lambda}\right\}\right).
\]
This price is the cost that corresponds to multiplying the seller's quantile by $1/\lambda$. If the buyer's bid $b$ satisfies $b \geq p^b$, the mechanism allocates and charges the buyer $p^b$.
\item Buyer-Proposing (\postquantile): Given the buyer's value($v$) and the seller distribution $F$, define the buyer's value-quantile as $\tau = F(v)$. Let $q(v) \in (0, \tau)$ be a parameter. The mechanism posts a price to the seller of
\[p^s = F^{-1}(q(v)).\]
This price is the cost corresponding to a quantile $q$ that is less than the buyer's value-quantile $\tau$. If the seller's bid $s$ satisfies $s \leq p^s$, the mechanism allocates and pays the seller $p^s$.
\end{itemize}

\paragraph{Properties of the maximum-weight matching (MWM).}
We first establish the monotonicity properties of the allocation rule that selects the unique maximum-weight matching. 
We then use it to show that the allocation rule is ``stable,'' meaning that agents cannot change whom they are matched with in the MWM by altering their bids.

\begin{lemma}\label{lem:mwm_monotonicity}
The rule that selects the unique MWM $M^*$ is monotone:
\begin{enumerate}
\item If $(i,j) \in M^*$ for profile $(\bids, \sbids)$, then for any $b'_i > b_i$, $(i,j)$ is also in $M^*$ for profile $((b'_i, \mathbf{b}_{-i}), \mathbf{s})$.
\item If $(i,j) \in M^*$ for profile $(\bids, \sbids)$, then for any $s'_j < s_j$, $(i,j)$ is also in $M^*$ for profile $(\bids, (s'_j, \sbids_{-j}))$.
\end{enumerate}

\end{lemma}

\begin{proof}
Let $M_A = M^*(\bids, \sbids)$ be the unique MWM at $v_i$, with $(i,j) \in M_A$. Let $\Delta b = b'_i - b_i > 0$. 
Consider any feasible $M_B \neq M_A$. We show $M_A$ still wins against $M_B$ at $b'_i$.

Case 1: $M_B$ does not contain buyer $i$.
The weight of $M_A$ increases by $\Delta v$, while the weight of $M_B$ is unchanged. 
Since $M_A$ had a weight $\ge M_B$ before, it now has a strictly higher primary weight. $M_A$ wins.

Case 2: $M_B$ also contains buyer $i$. The weights of both $M_A$ and $M_B$ increase by the same amount, $\Delta v$. Their relative primary weights are unchanged. If $M_A$ had a strictly higher primary weight at $b_i$, it still does. $M_A$ wins. If $M_A$ was tied with $M_B$ on primary weight at $b_i$, they are still tied at $b'_i$. Since $M_A = M^*(\mathbf{b}, \mathbf{s})$, $M_A$ must have won the tie-break. Since the global order is fixed, $M_A$ wins the tie-break again. 

Since $M_A$ wins against all competitors, $M_A$ remains the unique MWM. The proof for the second bullet point ($s_j$) is symmetric.
\end{proof}

This lemma leads to a crucial insight: an agent's matching partner is stable.

\begin{lemma}\label{lem:buyerinauction}
Consider a given profile $(\mathbf{b}, \mathbf{s}, E)$ and the unique maximum-weight matching $M^*(\mathbf{b}, \mathbf{s})$.
\begin{itemize}
    \item Consider any buyer $i$ that is matched in $M^*$, and let $j^*$ be the seller matched with $i$. For any other reported value $b'_{i}$, if buyer $i$ is matched in the new MWM $M^*( (b'_{i}, \mathbf{b}_{-i}), \mathbf{s})$, then $i$ must still be matched with $j^*$.
    \item Consider any seller $j$ that is matched in $M^*$, and let $i^*$ be the buyer matched with $j$. For any other reported cost $s'_{j}$, if seller $j$ is matched in the new MWM $M^*(\mathbf{b}, (s'_{j}, \mathbf{s}_{-j}))$, then $j$ must still be matched with $i^*$.
\end{itemize}
\end{lemma}
\begin{proof}
We prove the first bullet point. Let $M^*$ be the MWM for $b_i$, and $M'^*$ be the MWM for $b'_i$. We are given $M^*(i)=j^*$ and that $i$ is also matched in $M'^*(i)$. We must show $M'^*(i)=j^*$.

Case 1: $b'_i > b_i$.
    By \cref{lem:mwm_monotonicity}, since $(i,j^*) \in M^*$, it must be that $(i,j^*)$ is also in the new MWM, $M'^*$. As the MWM is unique, this forces $M'^*(i)=j^*$.

Case 2: $b'_i < b_i$.
    Let $M'^*(i)=j'$. We apply the logic in reverse, treating $M'^*$ as the starting matching and $b_i$ as the new, higher value. Since $(i,j') \in M'^*$, \cref{lem:mwm_monotonicity} requires that the partner for $i$ be preserved in the resulting MWM ($M^*$). This means $M^*(i)$ must also be $j'$. But we are given $M^*(i)=j^*$ and hence by uniqueness, it must be that $j' = j^*$.

In all cases, the partner must remain the same. The proof for the second bullet point (varying $s_j$) is symmetric.
\end{proof}

\begin{corollary}
Consider any $(i,j) \in E$ and any profile $(\bids, \sbids)$. It holds that
 \begin{equation}\label{eq:maxmatchingequivalence}
    (i,j) \in M^*(\mathbf{b}, \mathbf{s}) \iff s_j \le \bar{c}_{ij}(\mathbf{b}, \mathbf{s}_{-j}).
\end{equation}
\end{corollary}

\begin{proof}
We prove the equivalence by fixing seller $j$ and all other bids $(\mathbf{b}, \mathbf{s}_{-j})$.
First, by \cref{lem:buyerinauction}, $j$ can be matched to at most one buyer (call this unique potential partner $i^*$), regardless of his reported cost. This gives us two cases for any pair $(i,j)$:

Case 1: $i$ is the unique potential partner ($i = i^*$). By the monotonicity of the MWM, the allocation rule for this pair $(i,j)$ is a threshold rule based on $j$'s cost. By definition, $\bar{c}_{ij}$ is this threshold (the supremum of costs for which $(i,j) \in M^*$). Therefore, $(i,j) \in M^*$ if and only if $s_j \le \bar{c}_{ij}$. 

Case 2: $i$ is \emph{not} the unique potential partner ($i \neq i^*$). Consider the LHS. By \cref{lem:buyerinauction}, $j$ can never be matched with $i$. Thus, $\indicator[(i,j) \in M^*]$ is always $0$. Now consider the RHS. By definition, $\bar{c}_{ij}$ is the supremum of the set of costs for which $(i,j) \in M^*$. Since this set is empty, $\bar{c}_{ij} = \sup(\varnothing) = -\infty$. Therefore, $\indicator[s_j \le \bar{c}_{ij}]$ (i.e., $\indicator[s_j \le -\infty]$) is also always $0$. 
\end{proof}

\subsection{Incentive and Profit Guarantees for the Buyer-Run Auction \texorpdfstring{\sdma(B)}{MA(B)}}

We are now ready to prove the incentive and profit guarantees for the Meta-Auction framework when paired with an appropriate \btprice. We first prove the ``buyer-run'' case. Since we are maximizing the total expected buyers' profit, the buyers (the ``running'' side) are treated as non-strategic. The auctions, therefore, have access to the buyers' true values $v_i$ and do not need to consider their incentives. Consequently, our only task is to prove incentive compatibility for the strategic sellers. We show that the Meta-Auction (\sdma) run by the buyer, when paired with any \btprice\ rule that is DSIC for the seller (e.g., a posted price based on $v_i$ and $\bar{F}_{ij}$), is DSIC for all sellers in the matching market.
\begin{lemma}\label{lem:buyermonotone}
Given a \btprice$^B$\ that is truthful for the seller (DSIC-S), the allocation rule of \sdma(B){} is monotone for the sellers.
\end{lemma}
\begin{proof}
Let $x_j(\sbids)$ denote the allocation for seller $j$ given a seller cost--bidding profile $\sbids$. We must show that if $x_j(s_j, \sbids_{-j}) = 1$, then $x_j(s'_j, \sbids_{-j}) = 1$ for any $s'_j < s_j$, holding all other reports $\mathbf{v}$ and $\sbids_{-j}$ constant. Assume $x_j(s_j, \sbids_{-j}) = 1$. By the definition of the \sdma, two conditions were met:
\begin{enumerate}
    \item Seller $j$ was in the unique maximum-weight matching, $M^*$, with partner $i$.
    \item Seller $j$ was allocated by the bilateral pricing rule: $\btprice^B(\bar{F}_{ij}, s_j, v_i)$ returned an allocation.
\end{enumerate}

Now, consider the seller reporting $s'_j < s_j$. We check if both conditions are still met. By the monotonicity of the MWM (\cref{lem:mwm_monotonicity}), the same pair $(i,j)$ must also be in the new matching, $M'^*$. Thus, $M'^*(j) = i$.

In addition, the critical cost $\bar{c}_{ij}$ is a function of $(\mathbf{v}, \sbids_{-j})$ and is independent of $s_j$. Consequently, the modified distribution $\bar{F}_{ij}$, which defines the bilateral trade instance, is also independent of $s_j$. The seller faces the \emph{exact same} bilateral mechanism context regardless of their report.

Since $\btprice^B$ is truthful for the seller, its allocation rule must be monotone with respect to the seller's cost. We are given that $s'_j < s_j$, and that seller $j$ is allocated in $\btprice^B(\bar{F}_{ij}, s_j, v_i)$. By truthfulness (monotonicity) of \btprice, seller $j$ must also be allocated in $\btprice^B(\bar{F}_{ij}, s_j, v_i)$.
Both conditions are met, so $x_j(\sbids') = 1$. The allocation rule is therefore monotone.
\end{proof}
We now provide the profit guarantees. In particular, we lower-bound the total expected profit the buyers get from running \sdma(B) by the sum of the expected profit of running the corresponding \btprice$^B$\ on  each of the individual modified bilateral trade instances, in every realization.

\begin{lemma}\label{lem:buyerprofit}
Let the \sdma(B)\ be run by the buyer using a DSIC-S \btprice$^B$. Let $\pi^{i}_{\sdma(B)}(\mathbf{v}, \mathbf{s})$ 
be the profit of buyer $i$ from the \sdma(B) on the input $(\mathbf{v}, \mathbf{s})$. 
Let $\pi^{i}_{\btprice}(\bar{F}_{ij}(\val, \sbids_{-j}), s_j, v_i)$ be the profit buyer $i$ would get from the standalone \btprice$^B$\ run on the modified instance $(\bar{F}_{ij}, s_j, v_i)$. For any realization $(\mathbf{v}, \mathbf{s})$ such that $(i,j) \in M^*(\mathbf{v}, \mathbf{s})$:
\[
\pi^{i}_{\sdma(B)}(\mathbf{v}, \mathbf{s}) \geq \pi^{i}_{\btprice}(\bar{F}_{ij}(\val, \sbids_{-j}), s_j, v_i).
\]
\end{lemma}
\begin{proof}
    Let $(\mathbf{v}, \mathbf{s})$ be a fixed realization when $(i,j) \in M^*(\mathbf{v}, \mathbf{s})$.
    By the definition of this \sdma, the allocation $x_{ij}$ is determined by the \btprice\ subroutine, i.e.,
    \[x_{ij} = x^{\btprice}_{ij}(\bar{F}_{ij}, s_j, v_i).\]
     Let $p^{\btprice}_j$ be the seller's threshold price in the standalone \btprice\ (i.e., the max $s_j$ that wins). Let $p^s_j$ be the seller's threshold price in the full \sdma\ (i.e., the max $s_j$ that wins).
    Reporting any cost $s_j > p^{\btprice}_j$ cannot win the \btprice\ allocation. Since winning the \sdma\ requires winning the \btprice\ allocation, any cost $s_j > p^{\btprice}_j$ also cannot win the \sdma. This means that the \sdma's threshold price must be weakly less than the \btprice's threshold price:
    \[ p^s_j \le p^{\btprice}_j. \]
     Since $x_{i} = x^{\btprice}_{ij}$ and $p^s_j \le p^{\btprice}_j$, it follows that
     \[
     \pi^{i}_{\sdma(B)}(\mathbf{v}, \mathbf{s}) = x_{ij} \cdot (v_i - p^s_j) \geq  x^{\btprice}_{ij} \cdot (v_i - p^{\btprice}_j) =  \pi^{i}_{\btprice}(\bar{F}_{ij}, s_j, v_i). \qedhere
     \]
\end{proof}

\subsection{Incentive and Profit Guarantees for the Seller-Run Auction \texorpdfstring{\sdma(S)}{MA(S)}}
We now prove the buyers' incentive properties and the profit guarantees in the ``seller-run'' case. This case is more complex than the buyer-run case because a buyer's report has a dual role in the modified bilateral trade instance: it not only serves as their bid but also endogenously changes the modified distribution, $\bar{F}$. To handle this complexity, we must impose a mild requirement on the \btprice$^S$\ mechanisms that the seller can run. We then show that this requirement is satisfied by the \doublequantile\ family studied in the bilateral trade literature. We first state this desired property.
\begin{definition}[Cap-Monotonicity]\label{def:btproperty}
Let $F$ be the original seller's distribution, and let $F_\kappa(x) = F(\min(x, \kappa))$ be the distribution $F$ ``censored'' at $\kappa$ (i.e., any true cost $c$ with $c > \kappa$ is treated as $\infty$). A \btprice\ rule $x^{\btprice}$ is \emph{\property} if, for any two caps $\kappa_1 \le \kappa_2$ as well as any fixed buyer bid $b_i$ and seller cost $c_j$, it holds that
\[
x^{\btprice}(F_{\kappa_1}, b_i, c_j) = 1 \implies x^{\btprice}(F_{\kappa_2}, b_i, c_j) = 1.
\]
\end{definition}
{In simple words, cap-monotonicity means that if a trade is successful under a strict cost-limit, it must remain successful if that limit is loosened.
The ``cap'' represents the maximum cost the market can include a seller while still keeping the overall matching optimal. This property ensures that making the market more flexible, by increasing the allowable cost for a seller, never accidentally breaks a trade that was already allocating. }

The cap-monotonicity property induces the following price relations of a \btprice\ across instances.

\begin{lemma}\label{lem:threshold_monotonicity}Let a \btprice\ be both \property\ and DSIC-B. For a fixed seller bid $c_j$, let $T(\kappa_1, c_j)$ and $T(\kappa_2, c_j)$ be the respective threshold prices (lowest bid to win for the buyer) for the \btprice\ games in environments $F_{\kappa_1}$ and $F_{\kappa_2}$, where $\kappa_1 \le \kappa_2$. Then, the threshold for the ``low cap'' game is (weakly) higher than the threshold for the ``high cap'' game:
\[
T(\kappa_1, c_j) \ge T(\kappa_2, c_j).
\]
\end{lemma}

\begin{proof}
Let $W_1 = \{ b' \mid x^{\btprice}(F_{\kappa_1}, b', c_j) = 1 \}$ and $W_2 = \{ b' \mid x^{\btprice}(F_{\kappa_2}, b', c_j) = 1 \}$ be the set of all winning bids for the $F_{\kappa_1}$ and $F_{\kappa_2}$ games, respectively.

Consider any bid $b \in W_1$. By definition, $x^{\btprice}(F_{\kappa_1}, b, c_j) = 1$.
Since $\kappa_1 \le \kappa_2$ and the \btprice\ rule is \property, the buyer must also win in the $F_{\kappa_2}$ instance with the same bid $b$. Thus, $x^{\btprice}(F_{\kappa_2}, b, c_j) = 1$, which implies $b \in W_2$. This implies that $W_1 \subseteq W_2$.

Since the \btprice\ is DSIC-B, its allocation rule is monotone in the buyer's bid. The threshold price $T(\kappa, c_j)$ is the infimum of the set of winning bids $W$.
Because $W_1 \subseteq W_2$, the infimum of the subset $W_1$ must be greater than or equal to the infimum of the superset $W_2$.
Therefore, $\inf(W_1) \ge \inf(W_2)$, which implies $T(\kappa_1, c_j) \ge T(\kappa_2, c_j)$.
\end{proof}

We now present how buyer's value changes the modified seller distribution $\bar{F}_{ij}$ in \sdma. In particular, we show that a higher bid induces a higher cap on the distribution.

\begin{lemma}
\label{lem:capmonotone}
Consider a fixed profile of all reports from other agents, $(\mathbf{b}_{-i}, \mathbf{c})$. Let $\bar{c}_{ij}(b_i)$ denote the critical cost for pair $(i,j)$ computed by the \sdma(S)\ as a function of buyer $i$'s reported value, $b_i$.
For any two reported values $b'_i \ge b_i$, it holds that
\[ \bar{c}_{ij}(b'_i) \ge \bar{c}_{ij}(b_i). \]
\end{lemma}

\begin{proof}
Let $b'_i \ge b_i$. Let $c^* = \bar{c}_{ij}(b_i)$ be the critical cost corresponding to the lower value, $b_i$. 
Consider the profile $P = ((b_i, \mathbf{b}_{-i}), (c^*, \mathbf{c}_{-j}))$. By the definition of critical cost (\cref{def:criticalcost}), the pair $(i,j)$ is in the maximum-weight matching, $M^*$, for this profile $P$.

Now, consider a new profile $P'$ where only the buyer's value is increased from $b_i$ to $b'_i$, i.e., $P' = ((b'_i, \mathbf{b}_{-i}), (c^*, \mathbf{c}_{-j}))$. By \cref{lem:mwm_monotonicity}, since $(i,j) \in M^*$ at profile $P$, it must also be that $(i,j) \in M'^*$ at profile $P'$.

The new critical cost $\bar{c}_{ij}(b'_i)$ is the \emph{supremum} of all such costs for which $(i,j)$ is in the matching (given $b'_i$). Therefore, $\bar{c}_{ij}(b'_i)$ must be at least as large as $c^*$. Substituting $c^* = \bar{c}_{ij}(b_i)$, we have
\[ \bar{c}_{ij}(b'_i) \ge \bar{c}_{ij}(b_i). \qedhere\]
\end{proof}





\begin{lemma}\label{lem:sellermonotoncity}Given a \btprice$^S$\ that is \property\ (\cref{def:btproperty}) and DSIC-B (i.e., monotone in the buyer's bid for any fixed seller distribution), the allocation rule of \sdma(S)\ is monotone for the buyers.
\end{lemma}

\begin{proof}Let $x_i(\bids)$ denote the allocation for buyer $i$ given a value profile $\bids$. We must show that if $x_i(b_i, \bids_{-i}) = 1$, then $x_i(b'_i, \bids_{-i}) = 1$ for any $b'_i > b_i$, fixing s $\mathbf{c}$ and $\bids_{-i}$. Assume $x_i(b_i, \bids_{-i}) = 1$. By the definition of the \sdma(S), two conditions were met: (1) Buyer $i$ was in the unique maximum-weight matching, $M^*$, with partner $j$ (i.e., $(i,j) \in M^*(b_i, \mathbf{c})$). (2) Buyer $i$ was allocated by the bilateral allocation rule $\btprice^S(\bar{F}_{ij}(b_i), b_i, c_j)$.

Consider the buyer reporting $b'_i > b_i$. By the monotonicity of the MWM (\cref{lem:mwm_monotonicity}), since $(i,j) \in M^*(b_i, \mathbf{s})$ and $b'_i > b_i$, the same pair $(i,j)$ must also be in the new matching, $M^*(b'_i, \mathbf{c})$.
In addition,
by \cref{lem:capmonotone}, we have $\bar{c}_{ij}(b'_i) \ge \bar{c}_{ij}(b_i)$, i.e., a higher bid $b'_i$ induces a higher cap on the distribution.
Since the \btprice$^S$\ is \property\ (\cref{def:btproperty}), and buyer $i$ was allocated by the bilateral allocation rule $\btprice^S(\bar{F}_{ij}(b_i), b_i, c_j)$ (using the notation $F_{\bar{c}_{ij}(b_i)} = \bar{F}_{ij}(b_i)$), we therefore have buyer $i$ is also allocated by $\btprice^S(\bar{F}_{ij}(b'_i), b_i, c_j)$. Finally, since the \btprice$^S$\ is also DSIC-B, and $b'_i > b_i$, we have buyer $i$ is also allocated by $\btprice^S(\bar{F}_{ij}(b'_i), b'_i, c_j)$. Since both conditions are met, the buyer will be allocated at bid $b'_i$.
\end{proof}

We now show that the family \doublequantile\ is \property. 

\begin{lemma}\label{lem:doublequantileiscapmonotone}
The \doublequantile\ bilateral trade rule is \property\ (\cref{def:btproperty}).
\end{lemma}

\begin{proof}
Let $p^b(\kappa) = F_{\kappa}^{-1}(\min\{1, F_{\kappa}(c_j)/\lambda\})$ be the price posted by \doublequantile\ for a distribution capped at $\kappa$.
Assume $x^{\btprice}(F_{\kappa_1}, b_i, c_j) = 1$. This means that $b_i \ge p^b(\kappa_1)$, and in particular that $p^b(\kappa_1)$ is finite.

 $p^b(\kappa_1) < \infty$ further implies $c_j < \kappa_1$. To see this, suppose by contradiction that $c_j \ge \kappa_1$. Then $F_{\kappa_1}(c_j) = F(\kappa_1)$. The target quantile would be $F(\kappa_1)/\lambda > F(\kappa_1)$ (since $\lambda < 1$), making $p^b(\kappa_1) = \infty$, a contradiction.

Since $c_j < \kappa_1$ and $\kappa_1 \le \kappa_2$, we have $c_j < \kappa_2$.
This means 
$$F_{\kappa_1}(c_j) = F(\min(c_j, \kappa_1)) = F(c_j) = F(\min(c_j, \kappa_2))  = F_{\kappa_2}(c_j).$$
The target quantiles are therefore identical: $q^* = \min\{1, F(c_j)/\lambda\}$.

 $p^b(\kappa_1) < \infty$ also implies $q^* \le F(\kappa_1)$.
Since $\kappa_1 \le \kappa_2$, we have $F(\kappa_1) \le F(\kappa_2)$.
Thus, $q^* \le F(\kappa_1) \le F(\kappa_2)$, together with the definition of $F_{\kappa}$ we have $F(q^*) = F_{\kappa_2}(q^*) = F_{\kappa_1}(q^*)$
Both inverse functions evaluate on the original $F$: 
$p^b(\kappa_1) = F^{-1}(q^*)$ and $p^b(\kappa_2) = F^{-1}(q^*)$. Thus the posted prices are identical.
Given $b_i \ge p^b(\kappa_1)$, we also have $b_i \ge p^b(\kappa_2)$, which implies $x^{\btprice}(F_{\kappa_2}, b_i, c_j) = 1$.
\end{proof}

Finally we provide the profit guarantee. As before, we lower-bound the total expected profit the seller gets from running the Meta-Auction with a \property\ \btprice$^S$\ by the sum of the profit of each of the individual modified instances, in every realization.
\begin{lemma} \label{lem:sellerprofit}
Let the \sdma(S)\ be run by the seller, using a \btprice$^S$\ that is \property\ and DSIC-B (e.g., \doublequantile), to determine the allocation $x_{ij}$. 
Let $p^b_i$ be the buyer's threshold price in the \sdma(S)\ and  $\pi^j_{\sdma(S)}(\mathbf{b}, \mathbf{c})$ be the profit of seller $j$ from trading with $i$ in this \sdma(S). 
Let $\pi^j_{\btprice}(\bar{F}_{ij}(\mathbf{b}, \mathbf{c}_{-j}), b_i, c_j)$ be the profit seller $j$ would get from the standalone \btprice$^S$, priced at its own threshold $p^{\btprice}_i$.

For any realization $(\mathbf{b}, \mathbf{c})$ where $(i,j) \in M^*(\mathbf{b}, \mathbf{c})$, we have
\[ \pi^j_{\sdma(S)}(\mathbf{b}, \mathbf{c}) \ge \pi^j_{\btprice}(\bar{F}_{ij}(\mathbf{b}, \mathbf{c}_{-j}), b_i, c_j). \]
\end{lemma}
\begin{proof}
Let $p_i^{\sdma}$ be the buyer's threshold price in the \sdma(S), and $p_i^{\btprice}$ be the buyer's threshold price in the \btprice$^S$. The seller's profit is $x \cdot (p - c_j)$.
For a realization $(v,c)$ where $(i,j) \in M^{*}(v,c)$, the \sdma(S) Phase 1 condition is met. Thus, the \sdma{} allocation $x_{ij}^{\sdma}$ is determined only by the Phase 2 \btprice$^S$ subroutine, making it identical to the standalone \btprice$^S$ allocation: $$x_{ij}^{\sdma}(v,c) = x^{\btprice}(\overline{F}_{ij}(\cdot), v_i, c_j) = x_{ij}^{\btprice}(v,c).$$
Let this be $x_{ij}$.
The profits are $\pi_{\sdma(S)}^{j} = x_{ij} \cdot (p_i^{\sdma} - c_j)$ and $\pi_{\btprice}^{j} = x_{ij} \cdot (p_i^{\btprice} - c_j)$.
If $x_{ij} = 0$, the profits are zero and the inequality holds. If $x_{ij} = 1$, we must show $p_i^{\sdma} \ge p_i^{\btprice}$.

Let $\kappa(b_i) = \overline{c}_{ij}(b_i, \bids_{-i}, c_{-j})$ be the critical cost cap induced by bid $b_i$. Let $T(\kappa, c_j)$ be the buyer's threshold for a \btprice$^{S}$ game with cap $\kappa$.
The standalone \btprice$^{S}$ threshold is $p_i^{\btprice} = T(\kappa(v_i), c_j)$.
The \sdma(S) threshold $p_i^{\sdma}$ is the minimum bid $b_i$ that wins both phases: (1) $(i,j) \in M^{*}$ and (2) $b_i \ge T(\kappa(b_i), c_j)$.

Since $x_{ij}=1$, the true value $v_i$ wins the standalone \btprice$^{S}$, so $v_i \ge p_i^{\btprice}$.
Consider any bid $b_i' < p_i^{\btprice}$. This implies $b_i' < v_i$.
By \cref{lem:capmonotone}, a lower bid induces a (weakly) smaller cap: $\kappa(b_i') \le \kappa(v_i)$.
By \cref{lem:threshold_monotonicity}, a smaller cap induces a (weakly) higher threshold: $T(\kappa(b_i'), c_j) \ge T(\kappa(v_i), c_j)$.
Chaining these facts, we have: $$b_i' < p_i^{\btprice} = T(\kappa(v_i), c_j) \le T(\kappa(b_i'), c_j).$$
This shows $b_i' < T(\kappa(b_i'), c_j)$, meaning the bid $b_i'$ would fails the induced \sdma's Phase 2 condition.
Since any bid $b_i' < p_i^{\btprice}$ is a losing bid in the \sdma(S), the \sdma(S) threshold $p_i^{\sdma}$ must be at least $p_i^{\btprice}$.
Thus, $\pi_{\sdma(S)}^{j} \ge \pi_{\btprice}^{j}$.
\end{proof}




\section{{Proof of the Main Result via} GFT Decomposition}\label{sec:gft_decomposition}
In this section, we put everything together and prove our main result: an approximation of the optimal expected GFT. A central challenge in analyzing the welfare properties of complex matching markets is the interdependency between all agents. The gains-from-trade (GFT) generated by one pair $(i,j)$ is conditional on that pair being part of the first-best matching $M^*$, which in turn depends on the values and costs of all other agents.

The following lemma provides an analytical tool to bridge this gap. We demonstrate that the total expected GFT of the entire market can be additively decomposed into a sum. Each term in this sum corresponds to the expected GFT of a \emph{modified bilateral trade instance} between a pair $(i,j)$. This decomposition is the key to our proof. Crucially, these modified instances are the exact same instances (defined by the modified distributions $\bar{F}_{ij}$) that serve as the input to the \btprice\ subroutines in our Meta-Auction framework. This allows for the direct application of profit--GFT guarantees from the well-understood bilateral trade literature to our broader market setting.

\begin{definition}[Total Expected GFT]\label{def:gft}
\[ \gft \equiv \E_{\mathbf{v},\mathbf{c}}\left[\sum_{(i,j) \in M^*(\mathbf{v}, \mathbf{c})} (v_i - c_j)\right]. \]
\end{definition}

\begin{definition}[GFT of the Modified Bilateral Trade]\label{def:btinstance}
We define the optimal GFT from the \emph{modified bilateral trade instance} for a buyer $i$ with realized value $v_i$ and seller $j$'s cost drawn from distribution $\bar{F}_{ij}(\mathbf{v}, \mathbf{c}_{-j})$ as $\gft_{i,j}(\mathbf{v}, \mathbf{c}_{-j})$
\[
\gft_{i,j}(\mathbf{v}, \mathbf{c}_{-j}) \equiv \int_{q=0}^{1} \left(v_i - \bar{F}_{ij}^{-1}(q \mid \mathbf{v}, \mathbf{c}_{-j})\right)^+ \dif q = \int_{q=0}^{\bar{q}_{ij}(\mathbf{v}, \mathbf{c}_{-j})} \left(v_i - \bar{F}_{ij}^{-1}(q \mid \mathbf{v}, \mathbf{c}_{-j})\right) \dif q,
\]
where $\bar{q}_{ij}(\cdot)$ is the quantile of $\bar{c}_{ij}(\cdot)$ with respect to $F_j$. The equality comes from the fact that for any $q > \bar{q}_{ij}(\mathbf{v}, \mathbf{c}_{-j})$, we have $\bar{F}_{ij}(q) = \infty$ and $v_i \geq \bar{q}_{ij}(\mathbf{v}, \mathbf{c}_{-j})$ by the definition of $\bar{q}$.
\end{definition}

\begin{lemma}\label{lem:decomposition}
The total expected GFT, denoted $\gft$, is equivalent to the sum over all edges $E$ of the expected GFT from the \emph{modified bilateral trade instances}:
\[
\gft = \sum_{(i,j) \in E} \E_{\mathbf{v}, \mathbf{c}_{-j}} \left[ \gft_{i,j}(\mathbf{v}, \mathbf{c}_{-j}) \right].
\]
\end{lemma}

\begin{proof}
We start with the definition of $\gft$ (\cref{def:gft}) and apply linearity of expectation:
\[
    \gft = \sum_{(i,j) \in E} \E_{\mathbf{v},\mathbf{c}}\left[ (v_i - c_j) \cdot \indicator[(i,j) \in M^*(\mathbf{v}, \mathbf{c})] \right].
\]
Now we apply the law of total expectation, conditioning on $(\mathbf{v}, \mathbf{c}_{-j})$ and taking the inner expectation over $c_j \sim F_j$:
\begin{equation}\label{eq:gftmain}
    \gft = \sum_{(i,j) \in E} \E_{\mathbf{v}, \mathbf{c}_{-j}} \left[ \E_{c_j \sim F_j} [ (v_i - c_j) \cdot \indicator[(i,j) \in M^*(\mathbf{v}, \mathbf{c})] \mid \mathbf{v}, \mathbf{c}_{-j} ] \right].
\end{equation}

Now we analyze the inner expectation, which we will call $I$. This term $I$ represents the \emph{actual} contribution of pair $(i,j)$ to the total GFT, given $(\mathbf{v}, \mathbf{c}_{-j})$.
\[
I = \E_{c_j \sim F_j} [ (v_i - c_j) \cdot \indicator[(i,j) \in M^*(\mathbf{v}, \mathbf{c})] \mid \mathbf{v}, \mathbf{c}_{-j} ].
\]
By \cref{eq:maxmatchingequivalence} we know that given $(\mathbf{v}, \mathbf{c}_{-j})$, the indicator is equivalent to $\indicator[c_j \le \bar{c}_{ij}(\mathbf{v}, \mathbf{c}_{-j})]$. Substitute this into $I$, which allows us to write the expectation as an integral with respect to the quantile, 
Let $\bar{q}_{ij}(\mathbf{v}, \mathbf{c}_{-j}) = F_j(\bar{c}_{ij}(\mathbf{v}, \mathbf{c}_{-j}))$.
\[
I = \int_{q=0}^{\bar{q}_{ij}(\mathbf{v}, \mathbf{c}_{-j})} \left(v_i - F_j^{-1}(q)\right) \dif q.
\]
Recall the definition of \emph{conditional modified distribution} (\cref{def:modified_distribution}). For any $x \le \bar{c}_{ij}(\mathbf{v}, \mathbf{c}_{-j})$, we have $\bar{F}_{ij}(x \mid \cdot) = F_j(x)$. This implies that for any quantile $q \in [0, \bar{q}_{ij}(\mathbf{v}, \mathbf{c}_{-j})]$, the inverse is identical: $\bar{F}_{ij}^{-1}(q \mid \cdot) = F_j^{-1}(q)$. We can therefore substitute $\bar{F}_{ij}^{-1}$ into the integral above without changing its value:
    \[
    I = \int_{q=0}^{\bar{q}_{ij}(\mathbf{v}, \mathbf{c}_{-j})} \left(v_i - \bar{F}_{ij}^{-1}(q \mid \mathbf{v}, \mathbf{c}_{-j})\right) \dif q.
    \]
    By \emph{GFT of the modified bilateral trade} (\cref{def:btinstance}),
    this integral $I$ is precisely $\gft_{i,j}(\mathbf{v}, \mathbf{c}_{-j})$.
    Finally, substitute this result ($I = \gft_{i,j}(\mathbf{v}, \mathbf{c}_{-j})$) back into \cref{eq:gftmain}:
    \[ \gft = \sum_{(i,j) \in E} \E_{\mathbf{v}, \mathbf{c}_{-j}} \left[ \gft_{i,j}(\mathbf{v}, \mathbf{c}_{-j}) \right]. \]
This proves that the term $I$, which is the \emph{actual} contribution of pair $(i,j)$ to the total GFT, is \emph{equivalent} to the GFT of the modified bilateral trade instance we defined in \cref{def:btinstance}.
\end{proof}

We now state a result from the bilateral trade literature that approximates the GFT with the profits of both agents running some \btprice\ auction. In particular, there exists a \doublequantile\ that the seller can run and a \postquantile\ that the buyer can run, that together achieve a 3.15-approximation of the expected GFT, where the expectation is taken over the seller's cost, and holding the buyer's value fixed at $v$.
\begin{theorem}[\cite{DBLP:conf/wine/Fei22,hw25}]\footnote{We note that the analysis in \cite{DBLP:conf/wine/Fei22} is conducted with respect to the buyer's value distribution, fixing the seller's cost. In contrast, our setup fixes the buyer's value and takes the expectation over the seller's cost distribution. However, the two approaches are equivalent due to the inherent symmetry of the bilateral trade problem, as also noted in \cite[Remark 1]{DBLP:conf/wine/Fei22}.}\label{thm:bilateral_approx}
Given a seller's distribution $F$ and a buyer's value $v$, let $\pi_s(F,v,c)$ be the seller's realized profit for a cost $c$ from running \doublequantile\ with $\lambda \approx 0.317844$, and let $\pi_b(F, v, c)$ be the buyer's realized profit for a cost $c$ from running the ex-ante total-expected-profit-maximizing auction from \postquantile. It holds that
\[
\E_{c \sim F} \left[\frac{1}{2} \left(\pi_s(F,v,c) +  \pi_b(F, v, c)\right)\right] \geq \frac{1}{3.15} \int_{q = 0}^{1} \left(v- F^{-1}(q)\right)^+ \dif q.
\]
\end{theorem}

We now prove our main theorem (\cref{thm:main_profit_technical}), that randomizing between the Generalized Sellers-Offering Mechanism (GSOM) and the Generalized Buyers-Offering Mechanism (GBOM) yields a constant-factor approximation of the total expected gains-from-trade. Our proof follows a profit-based route, establishing a chain of inequalities. We first lower-bound the optimal total expected profits, $\Pi_S(\text{GSOM})$ and $\Pi_B(\text{GBOM})$, with the total expected profits of a specific feasible mechanism, the $\sdma$ auction. The profit of the $\sdma$ is then decomposed into a sum of profits on the individual edges of the matching $M^*$. We then apply our prior results (\cref{lem:buyerprofit,lem:sellerprofit}) to further lower-bound these per-edge $\sdma$ profits by the profits of the corresponding bilateral trade (\btprice) auctions. This reduction is the key, as it allows us to invoke the 3.15-approximation from the bilateral trade setting (\cref{thm:bilateral_approx}), which connects the randomized profit of a \btprice{} instance to its GFT. Finally, by applying the Law of Total Expectation and the GFT decomposition from \cref{lem:decomposition}, we sum these per-edge GFT components to recover the total expected GFT, which completes the bound.

\begin{proof}[Proof of \cref{thm:main_profit_technical}]
One feasible auction for the seller to run is $\sdma(\doublequantile)$. \doublequantile\ is \property\ by \cref{lem:doublequantileiscapmonotone}, \
and thus \sdma(\doublequantile) is truthful for the buyer by \cref{lem:sellermonotoncity}. Similarly, the buyer can choose to run $\sdma(\postquantile)$ that is truthful for the sellers by \cref{lem:buyermonotone}. Let $\pi_{\sdma(S)}(\val, \cost)$ and $\pi_{\sdma(B)}(\val, \cost)$ be the \emph{total realized profits} for the seller and buyer running the \sdma\ auction with the previously stated \btprice\ given a specific profile $(\val, \cost)$. Similarly let $\pi^j_{\btprice}$ and $\pi^i_{\btprice}$ be the profit of the bilateral trade auctions on edge $(i,j)$.
The total expected profit is the expectation of these realized profits, i.e., $\Pi_S(\sdma) = \E_{\val,\cost}[\pi_{\sdma(S)}(\val, \cost)]$.

We begin with the randomized optimal total expected profit and by \cref{obs:optimalprofit}, we have that GSOM (respectively, GBOM) maximizes the expected profit of the sellers (respectively, buyers) among all BIC--IR mechanisms. 
{\allowdisplaybreaks
\begin{align*}
&\quad\ \frac{1}{2} \Pi_S(\text{GSOM}) + \frac{1}{2} \Pi_B(\text{GBOM})\\
&\ge \frac{1}{2} \Pi_S(\sdma) + \frac{1}{2} \Pi_B(\sdma)\tag{\cref{obs:optimalprofit}} \\
&= \frac{1}{2} \E_{\val,\cost}\left[ \pi_{\sdma(S)}(\val, \cost) \right] + \frac{1}{2} \E_{\val,\cost}\left[ \pi_{\sdma(B)}(\val, \cost) \right]\tag{Definition of $\Pi$} \\
&= \frac{1}{2} \E_{\val,\cost}\left[ \sum_{(i,j) \in E} \pi^j_{\sdma(S)}(\val, \cost) \cdot \indicator[(i,j) \in M^*] \right] + \frac{1}{2} \E_{\val,\cost}\left[ \sum_{(i,j) \in E} \pi^i_{\sdma(B)}(\val, \cost) \cdot \indicator[(i,j) \in M^*] \right] \\
&\ge  \frac{1}{2} \E_{\val,\cost}\left[ \sum_{(i,j) \in E}  \pi^j_{\btprice}(\bar{F}_{ij}(\val, \cost_{-j}), v_i, c_j) \cdot \indicator[(i,j) \in M^*] \right] \\
&\quad + \frac{1}{2} \E_{\val,\cost}\left[ \sum_{(i,j) \in E}  \pi^i_{\btprice}(\bar{F}_{ij}(\val, \cost_{-j}), v_i, c_j) \cdot \indicator[(i,j) \in M^*] \right] \tag{\cref{lem:buyerprofit,lem:sellerprofit}}\\
&= \sum_{(i,j) \in E} \E_{\val,\cost}\left[ \frac{1}{2}\left( \pi^j_{\btprice}(\bar{F}_{ij}(\val, \cost_{-j}), v_i, c_j) + \pi^i_{\btprice}(\bar{F}_{ij}(\val, \cost_{-j}), v_i, c_j) \right) \cdot \indicator[(i,j) \in M^*] \right] \\
&= \sum_{(i,j) \in E} \E_{\val,\cost}\left[ \frac{1}{2}\left( \pi^j_{\btprice}(\bar{F}_{ij}(\val, \cost_{-j}), v_i, c_j) + \pi^i_{\btprice}(\bar{F}_{ij}(\val, \cost_{-j}), v_i, c_j) \right) \cdot \indicator[c_j \leq \bar{c}_{ij}(\val, \cost_{-j})] \right] \tag{\cref{eq:maxmatchingequivalence}}\\
&= \sum_{(i,j) \in E} \E_{\val,\cost_{-j}} \left[ \E_{c_j \sim F_j} \left[ \frac{1}{2}\left( \pi^j_{\btprice}(\bar{F}_{ij}(\val, \cost_{-j}), v_i, c_j) + \pi^i_{\btprice}(\bar{F}_{ij}(\val, \cost_{-j}), v_i, c_j) \right) \cdot \indicator[c_j \le \bar{c}_{ij}( \val,\cost_{-j}) \right] \right] \\
&= \sum_{(i,j) \in E} \E_{\val,\cost_{-j}} \left[ \E_{c_j \sim \bar{F}_{ij}(\val,\cost_{-j})} \left[ \frac{1}{2}\left( \pi^j_{\btprice}(\bar{F}_{ij}(\val, \cost_{-j}), v_i, c_j) + \pi^i_{\btprice}(\bar{F}_{ij}(\val, \cost_{-j}), v_i, c_j) \right) \right] \right]  \tag{\cref{def:modified_distribution}}\\
&\ge \sum_{(i,j) \in E} \E_{\val,\cost_{-j}} \left[ \frac{1}{3.15} \int_{q = 0}^{1} \left(v- \bar{F}_{ij}^{-1}(q \mid \val, \cost_{-j})\right)^+ \dif q \right]  \tag{\cref{thm:bilateral_approx}}\\
&= \sum_{(i,j) \in E} \E_{\val,\cost_{-j}} \left[ \frac{1}{3.15} \gft_{i,j}(\val, \cost_{-j}) \right]\tag{\cref{def:btinstance}} \\
&= \frac{1}{3.15} \left( \sum_{(i,j) \in E} \E_{\val,\cost_{-j}} \left[ \gft_{i,j}(\val, \cost_{-j}) \right] \right) = \frac{1}{3.15} \gft. \tag{\cref{lem:decomposition}}
\end{align*}}
This completes our proof.
\end{proof}

\section{Generalization to Multi-Dimensional Markets}\label{sec:multi-dim}
In this section, we state the implications of our results in the market with $n$ unit-supply sellers and one multi-dimensional unit-demand buyer who has different values for different items.

\cite{CaiGMZ21} establish a formal reduction that connects the GFT approximation problem in the multi-dimensional unit-demand setting to the gap between the first-best ($\mathtt{FB}\text{-}\mathtt{GFT}^\mathtt{SD}$) and second-best ($\mathtt{SB}\text{-}\mathtt{GFT}^\mathtt{SD}$) GFT in the general single-dimensional matching market. The result of \cite[Theorem 4]{CaiGMZ21} proves that if this single-dimensional gap is bounded by a constant factor $c$ (i.e., $\mathtt{SB}\text{-}\mathtt{GFT}^\mathtt{SD} \ge \mathtt{FB}\text{-}\mathtt{GFT}^\mathtt{SD} / c$), then a specific mechanism they propose achieves a $1/(2c)$-approximation of the first-best GFT in the multi-dimensional unit-demand setting.

Our main result provides this missing bound $c$. The mechanism we analyze (randomizing between GSOM and GBOM) is DSIC, IR, and ex-ante WBB. By definition, the GFT it achieves is a lower bound on $\mathtt{SB}\text{-}\mathtt{GFT}^\mathtt{SD}$. Furthermore, because the mechanism is WBB, its expected GFT is at least its total expected profit:
\[
\mathtt{SB}\text{-}\mathtt{GFT}^\mathtt{SD} \ge \E[\mathrm{GFT}(\text{Our Mechanism})] \ge \frac{1}{2} \Pi_S(\text{GSOM}) + \frac{1}{2} \Pi_B(\text{GBOM}).
\]
Our \cref{thm:main_profit_technical} proves that this total expected profit is at least $\gft^* / 3.15$. Therefore, in the single-dimensional matching market, we prove $\mathtt{SB}\text{-}\mathtt{GFT}^\mathtt{SD} \ge \gft^* / 3.15$, establishing that the gap $c$ is at most $3.15$. Applying this constant to the reduction in \cite{CaiGMZ21} yields the following corollary.

\begin{corollary}
The mechanism proposed in \cite{CaiGMZ21} (the better of ``Generalized Buyer Offering Mechanism'' and ``Seller Adjusted Posted Price'') achieves at least a $1/6.3$-approximation of the first-best GFT in the matching market setting with one multi-dimensional unit-demand buyer and $n$ unit-supply sellers.
\end{corollary}

\bibliographystyle{alpha}
\bibliography{biblio}

\appendix
\section{GSOM and GBOM}\label{app:gsom_equivalence}
In this section, we consider the problem of maximizing profit for one side of the market. We leverage this fundamental connection to derive the profit-optimal auctions for each side, which operate by maximizing the corresponding expected virtual surplus. We identify these auctions as the GSOM (Generalized Sellers-Offering Mechanism) which maximizes sellers' profit, and the GBOM (Generalized Buyers-Offering Mechanism) which maximizes buyers' profit. These auctions were proposed and studied by \cite{BCWZ17}, but their optimality in terms of profit maximization was not explicitly stated.

A crucial property also shown by \cite{BCWZ17} is that these mechanisms are DSIC for both sides---not just the side from which profit is extracted---and ex-ante WBB.
One way to interpret this is that the optimal profits from the buyers can be shared among the winning sellers in a truthful way.

\begin{definition}[Buyer Virtual Value]
    The \emph{Myerson virtual value} $\varphi_i(b_i)$ for a buyer $i$ with value $b_i$ drawn from distribution $D_i$ (with CDF $D_i$ and PDF $d_i$) is $\varphi_i(b_i) = b_i - \frac{1-D_i(b_i)}{d_i(b_i)}$. Let $\tilde{\varphi}_i(b_i)$ be the \emph{ironed virtual value}, which is the monotonized version of $\varphi_i(b_i)$.
\end{definition}

\begin{definition}[Seller Virtual Value]
    The \emph{Myerson virtual value} (or \emph{virtual cost}) $\psi_j(s_j)$ for a seller $j$ with cost $s_j$ drawn from distribution $F_j$ (with CDF $F_j$ and PDF $f_j$) is $\psi_j(s_j) = s_j + \frac{F_j(s_j)}{f_j(s_j)}$. Let $\tilde{\psi}_j(s_j)$ be the \emph{ironed virtual value}, which is the monotonized version of $\psi_j(s_j)$.
\end{definition}

\begin{definition}[Virtual Surplus]
    Given buyer values $\mathbf{v}$ (with virtual values $\boldsymbol{\varphi}(\mathbf{v})$) and seller costs $\mathbf{c}$ (with virtual costs $\boldsymbol{\psi}(\mathbf{c})$), the \emph{virtual surplus} of an allocation $\mathbf{x}$ (represented by a matching $A$) is defined from the perspective of the side whose profit is being maximized:
    \begin{itemize}
        \item for maximizing sellers' total profit (extracting from buyers): $\sum_{(i,j) \in M} (\varphi_i(v_i) - c_j)$;
        \item for maximizing buyers' total profit (extracting from sellers): $\sum_{(i,j) \in M} (v_i - \psi_j(c_j))$.
    \end{itemize}
    We use the ironed virtual values ($\tilde{\varphi}_i, \tilde{\psi}_j$) when the original virtual values are not monotone.
\end{definition}

Myerson's key insight relates the expected profit of any BIC and IR auction to the expected virtual surplus it generates.

\begin{lemma}[Myerson's Profit Equivalence \cite{Myerson81}]
    For any BIC 
    and IR auction:
    \begin{itemize}
        \item The expected total seller profit is equal to the expected virtual surplus $\E[\sum_{(i,j) \in M} (\tilde{\varphi}_i(b_i) - s_j)]$.
        \item The expected total buyer profit is equal to the expected virtual surplus $\E[\sum_{(i,j) \in M} (b_i - \tilde{\psi}_j(s_j))]$.
    \end{itemize}
    Therefore, an auction maximizes the expected profit for one side if and only if it maximizes the corresponding expected virtual surplus.
\end{lemma}

With these definitions, we derive the profit-maximizing auctions and identify them with auctions previously studied in the literature.

\subsection{Profit-Maximizing Auctions}


\begin{itemize}
    \item \textbf{The GSOM Auction (Generalized Sellers-Offering Mechanism) \cite{BCWZ17}.} Given reports $(\mathbf{b}, \mathbf{s})$, find a matching $M^* \in \mathcal{F}$ that maximizes the total virtual surplus (for expected total sellers'  profit):
        \[
        M^*(\mathbf{b}, \mathbf{s}) \in \argmax_{A \in \mathcal{F}} \sum_{(i,j)\in A}(\tilde{\varphi}_{i}(b_{i})-s_{j}).
        \]
The allocation $\mathbf{x}$ is set by $M^*$. Payments $\mathbf{p}^b$ and $\mathbf{p}^s$ are the unique threshold payments that ensure DSIC and IR.

    \item \textbf{The GBOM Auction (Generalized Buyers-Offering Mechanism) \cite{BCWZ17}.} Given reports $(\mathbf{b}, \mathbf{s})$, find a matching $M^* \in \mathcal{F}$ that maximizes the total virtual surplus (for expected total buyers' profit):
        \[
        M^*(\mathbf{b}, \mathbf{s}) \in \argmax_{A \in \mathcal{F}} \sum_{(i,j)\in A}(b_{i}-\tilde{\psi}_{j}(s_{j})).
        \]
The allocation $\mathbf{x}$ is set by $M^*$. Payments $\mathbf{p}^b$ and $\mathbf{p}^s$ are the unique threshold payments that ensure DSIC and IR.
\end{itemize}
While the GSOM and GBOM were introduced previously, their optimality in terms of expected total sellers' and buyers' profit maximization, respectively, was not stated explicitly.
\begin{observation}\label{obs:optimalprofit}
        \textbf{GSOM} is an IR and DSIC-B auction that maximizes the expected total profit of sellers. 
        Analogously, \textbf{GBOM} is an IR and DSIC-S auction that maximizes the expected total profit of buyers.
\end{observation}
In fact, \cite{BCWZ17} show that GSOM is DSIC for the sellers as well. One interpretation is that the optimal profit can be shared among the winning sellers in a truthful way.
\begin{theorem}[\cite{BCWZ17}]
    Both GSOM and GBOM are DSIC-B, DSIC-S, IR-B, IR-S, and ex-ante WBB.
\end{theorem}

\subsection{BIC versions of GSOM}




We next define a BIC variant of GSOM in which the buyers remain DSIC, while the sellers are only
BIC. The allocation rule is  the same as in GSOM; only the payments to the profit-receiving
side are changed. This changes makes the WBB property holds ex-post instead of ex-ante.

\paragraph{Virtual-surplus matching for GSOM.}
Let $M^\varphi(\mathbf b,\mathbf s)$ denote the unique matching selected by GSOM:
\[
M^\varphi(\mathbf b,\mathbf s)\in
\argmax_{A\in\mathcal F}\sum_{(i,j)\in A}\bigl(\tilde\varphi_i(b_i)-s_j\bigr),
\]
with ties broken according to the fixed global order.

Since each ironed virtual value $\tilde\varphi_i$ is monotone nondecreasing, the proofs of
Lemma~\ref{lem:mwm_monotonicity} and Lemma~\ref{lem:buyerinauction} apply verbatim after
replacing the weight $b_i$ of buyer $i$ by the transformed weight $\tilde\varphi_i(b_i)$.
Therefore, for every fixed $(\mathbf b_{-i},\mathbf s)$, if buyer $i$ is matched in
$M^\varphi((b_i,\mathbf b_{-i}),\mathbf s)$ for some $b_i$, then as only $b_i$ varies she can only
be matched to a unique seller; symmetrically, for every fixed $(\mathbf b,\mathbf s_{-j})$, if seller
$j$ is matched in $M^\varphi(\mathbf b,(s_j,\mathbf s_{-j}))$ for some $s_j$, then as only $s_j$
varies she can only be matched to a unique buyer.

For a matched edge $(i,j)\in M^\varphi(\mathbf b,\mathbf s)$, define the \emph{competition slack}
\[
\delta_{ij}^\varphi(\mathbf b_{-i},\mathbf s_{-j})
:=
\max_{\substack{A\in\mathcal F\\ i\notin A}}
\sum_{(k,\ell)\in A}\bigl(\tilde\varphi_k(b_k)-s_\ell\bigr)
-
\max_{\substack{A\in\mathcal F\\ A\cup\{(i,j)\}\in\mathcal F}}
\sum_{(k,\ell)\in A}\bigl(\tilde\varphi_k(b_k)-s_\ell\bigr),
\]
where $i\notin A$ means that buyer $i$ is unmatched in $A$. By construction,
$\delta_{ij}^\varphi(\mathbf b_{-i},\mathbf s_{-j})\ge 0$.

The best matching that contains $(i,j)$ has total weight
\[
\tilde\varphi_i(b_i)-s_j
+
\max_{\substack{A\in\mathcal F\\ A\cup\{(i,j)\}\in\mathcal F}}
\sum_{(k,\ell)\in A}\bigl(\tilde\varphi_k(b_k)-s_\ell\bigr),
\]
while the best matching that leaves buyer $i$ unmatched has total weight
\[
\max_{\substack{A\in\mathcal F\\ i\notin A}}
\sum_{(k,\ell)\in A}\bigl(\tilde\varphi_k(b_k)-s_\ell\bigr).
\]
Therefore we have,
\begin{equation}\label{eq:gsom-pair-threshold}
(i,j)\in M^\varphi(\mathbf b,\mathbf s)
\iff
\tilde\varphi_i(b_i)-s_j\ge \delta_{ij}^\varphi(\mathbf b_{-i},\mathbf s_{-j}).
\end{equation}
It follows that the buyer-threshold and seller-threshold in the original GSOM are
\begin{align}
\bar b_{ij}^\varphi(\mathbf b_{-i},\mathbf s)
:=
\inf\{z:(i,j)\in M^\varphi((z,\mathbf b_{-i}),\mathbf s)\}=
\inf\bigl\{z:\tilde\varphi_i(z)\ge s_j+\delta_{ij}^\varphi(\mathbf b_{-i},\mathbf s_{-j})\bigr\},
\label{eq:gsom-buyer-threshold}
\\[3mm]
\bar s_{ij}^\varphi((b_i,\mathbf b_{-i}),\mathbf s_{-j})
:=
\sup\{z:(i,j)\in M^\varphi((b_i,\mathbf b_{-i}),(z,\mathbf s_{-j}))\}
=
\tilde\varphi_i(b_i)-\delta_{ij}^\varphi(\mathbf b_{-i},\mathbf s_{-j}).
\label{eq:gsom-seller-threshold}
\end{align}

\begin{definition}[GSOM-BIC]\label{def:bgsom}
The mechanism \textbf{GSOM-BIC} is defined as follows.

\begin{enumerate}
    \item Compute the matching $M^\varphi(\mathbf b,\mathbf s)$ that maximizes
    $\sum_{(i,j)\in A}(\tilde\varphi_i(b_i)-s_j)$.
    
    \item If buyer $i$ is matched to seller $j$, charge buyer $i$ her threshold payment
    \[
    p_i^b(\mathbf b,\mathbf s)=\bar b_{ij}^\varphi(\mathbf b_{-i},\mathbf s).
    \]
    If buyer $i$ is unmatched, set $p_i^b(\mathbf b,\mathbf s)=0$.
    
    \item If seller $j$ is matched to buyer $i$, pay seller $j$ the conditional expectation of the
    seller-threshold payment of the original GSOM, conditioning on the event that $(i,j)$ is selected:
    \[
    p_j^s(\mathbf b,\mathbf s)
    :=
    \E_{b_i' \sim D_i}\!\left[
        \bar s_{ij}^\varphi((b_i',\mathbf b_{-i}),\mathbf s_{-j})
        \,\middle|\,
        (i,j)\in M^\varphi((b_i',\mathbf b_{-i}),\mathbf s)
    \right].
    \]
    If seller $j$ is unmatched, set $p_j^s(\mathbf b,\mathbf s)=0$.
\end{enumerate}
\end{definition}

The next lemma gives a closed form for the payment to a matched seller.

\begin{lemma}\label{lem:gsom-bic-closed-form}
Fix a profile $(\mathbf b,\mathbf s)$ and a matched edge $(i,j)\in M^\varphi(\mathbf b,\mathbf s)$.
Then
\[
p_j^s(\mathbf b,\mathbf s)
=
\bar b_{ij}^\varphi(\mathbf b_{-i},\mathbf s)
-
\delta_{ij}^\varphi(\mathbf b_{-i},\mathbf s_{-j}).
\]
Consequently,
\[
p_j^s(\mathbf b,\mathbf s)\le p_i^b(\mathbf b,\mathbf s).
\]
\end{lemma}

\begin{proof}
Let
$
r=\bar b_{ij}^\varphi(\mathbf b_{-i},\mathbf s)$ and $
\delta=\delta_{ij}^\varphi(\mathbf b_{-i},\mathbf s_{-j}).
$
By \eqref{eq:gsom-pair-threshold}, the event that $(i,j)$ is selected when only buyer $i$'s value
varies is  the event
\[
\tilde\varphi_i(b_i')\ge s_j+\delta,
\]
which is equivalent to $b_i'\ge r$ by \eqref{eq:gsom-buyer-threshold}. Moreover, by
\eqref{eq:gsom-seller-threshold},
\[
\bar s_{ij}^\varphi((b_i',\mathbf b_{-i}),\mathbf s_{-j})
=
\tilde\varphi_i(b_i')-\delta.
\]
Therefore,
\begin{align*}
p_j^s(\mathbf b,\mathbf s)
&=
\E\!\left[\tilde\varphi_i(b_i')-\delta \mid b_i'\ge r\right] = 
\frac{\E\!\left[\tilde\varphi_i(b_i')\mathbf 1[b_i'\ge r]\right]}{\Pr[b_i'\ge r]}
-\delta.
\end{align*}
Now consider the single-buyer posted-price mechanism that allocates if and only if $b_i'\ge r$ and charges
price $r$. By Myerson's payment identity for this monotone single-parameter allocation rule,
\[
\E\!\left[\tilde\varphi_i(b_i')\mathbf 1[b_i'\ge r]\right]
=
r\Pr[b_i'\ge r].
\]
Substituting this identity above gives
\[
p_j^s(\mathbf b,\mathbf s)=r-\delta
=
\bar b_{ij}^\varphi(\mathbf b_{-i},\mathbf s)-\delta_{ij}^\varphi(\mathbf b_{-i},\mathbf s_{-j}),
\]
as claimed. Since $\delta_{ij}^\varphi\ge 0$, the second statement follows immediately.
\end{proof}

\begin{theorem}\label{thm:gsom-bic}
GSOM-BIC is DSIC-B, BIC-S, IR-B, IR-S, and ex-post WBB.
Moreover, for every seller $j$, every report $s_j$, and every fixed reports of all other agents,
the interim expected payment of seller $j$ in GSOM-BIC is  the same as in the original GSOM.
Consequently, GSOM-BIC and GSOM induce the same ex-ante total sellers' utility.
\end{theorem}

\begin{proof}
\emph{DSIC-B and IR-B.}
The allocation rule of GSOM-BIC is  the GSOM allocation rule, hence it is monotone in each
buyer's bid. Each matched buyer is charged the corresponding threshold payment
$\bar b_{ij}^\varphi(\mathbf b_{-i},\mathbf s)$. Therefore, by Myerson's lemma,
GSOM-BIC is DSIC-B and IR-B.

\noindent
\emph{BIC-S and equality of interim expected payments with GSOM.}
Fix a seller $j$, a report $s_j$, and all other reports $(\mathbf b,\mathbf s_{-j})$.
If seller $j$ is never matched, the statement is trivial. Otherwise, by the virtual-surplus analogue of
Lemma~\ref{lem:buyerinauction}, there is a unique buyer $i$ such that seller $j$ can ever be matched,
and whenever she is matched it is to buyer $i$.

By definition of GSOM-BIC,
\begin{align*}
&\E_{b_i'\sim D_i}\!\left[
    p_j^s((b_i',\mathbf b_{-i}),(s_j,\mathbf s_{-j}))
    \cdot
    \mathbf 1\bigl[(i,j)\in M^\varphi((b_i',\mathbf b_{-i}),(s_j,\mathbf s_{-j}))\bigr]
\right]
\\
&\qquad =
\E_{b_i'\sim D_i}\!\left[
    \bar s_{ij}^\varphi((b_i',\mathbf b_{-i}),\mathbf s_{-j})
    \cdot
    \mathbf 1\bigl[(i,j)\in M^\varphi((b_i',\mathbf b_{-i}),(s_j,\mathbf s_{-j}))\bigr]
\right].
\end{align*}
The right-hand side is  the interim expected payment that seller $j$ receives in the original
GSOM under the same report $s_j$. Since the allocation rule is also the same, seller $j$'s interim
utility in GSOM-BIC is  the same as her interim utility in GSOM.

The original GSOM is DSIC-S, so truthful reporting maximizes that utility for every realization of the
other agents' reports. Therefore truthful reporting also maximizes the interim utility in GSOM-BIC,
which proves BIC-S. The same equality shows that the interim expected payment of seller $j$ is
identical in the two mechanisms, and hence the ex-ante total sellers' utility is also identical.

\noindent
\emph{IR-S.}
Fix a matched edge $(i,j)\in M^\varphi(\mathbf b,\mathbf s)$. By definition,
\[
p_j^s(\mathbf b,\mathbf s)
=
\E_{b_i'\sim D_i}\!\left[
\bar s_{ij}^\varphi((b_i',\mathbf b_{-i}),\mathbf s_{-j})
\,\middle|\,
(i,j)\in M^\varphi((b_i',\mathbf b_{-i}),\mathbf s)
\right].
\]
For every $b_i'$ in the conditioning event, seller $j$ wins in the original GSOM at cost $s_j$.
Hence the corresponding seller-threshold satisfies
\[
\bar s_{ij}^\varphi((b_i',\mathbf b_{-i}),\mathbf s_{-j})\ge s_j.
\]
Taking conditional expectations we have $
p_j^s(\mathbf b,\mathbf s)\ge s_j.$

\noindent
\emph{Ex-post WBB.}
Fix a matched edge $(i,j)\in M^\varphi(\mathbf b,\mathbf s)$. The pairwise inequality used in the
proof of ex-ante WBB for the original GSOM gives
\[
\E_{b_i'\sim D_i}\!\left[
\bar b_{ij}^\varphi(\mathbf b_{-i},\mathbf s)\,
\mathbf 1[(i,j)\in M^\varphi((b_i',\mathbf b_{-i}),\mathbf s)]
\right]
\ge
\E_{b_i'\sim D_i}\!\left[
\bar s_{ij}^\varphi((b_i',\mathbf b_{-i}),\mathbf s_{-j})\,
\mathbf 1[(i,j)\in M^\varphi((b_i',\mathbf b_{-i}),\mathbf s)]
\right].
\]
Since $\bar b_{ij}^\varphi(\mathbf b_{-i},\mathbf s)=p_i^b(\mathbf b,\mathbf s)$ is constant in $b_i'$,
dividing by the probability of the conditioning event gives
\[
p_i^b(\mathbf b,\mathbf s)\ge p_j^s(\mathbf b,\mathbf s).
\]
Summing over all matched edges proves
\[
\sum_i p_i^b(\mathbf b,\mathbf s)\ge \sum_j p_j^s(\mathbf b,\mathbf s). \qedhere
\]
\end{proof}

\paragraph{Remark.}
GSOM-BIC need not be ex-post strongly budget-balanced. Consider one buyer $i$, two sellers $j,k$, and feasible matchings
$\mathcal F=\{\emptyset,\{(i,j)\},\{(i,k)\}\}$. Let $B_i\sim U[0,1]$, so $\tilde\varphi_i(b)=\varphi_i(b)=2b-1$, and consider the realized profile
$b_i=1$, $s_j=0$, $s_k=\tfrac12$. Then GSOM-BIC matches $i$ to $j$. The buyer-threshold is
$\bar b_{ij}^\varphi=\tfrac12$, hence $p_i^b=\tfrac12$. In the original GSOM, for a hypothetical buyer report $b$, seller $j$'s threshold is
$\bar s_{ij}^\varphi(b,s_k)=\min\{\tfrac12,\,2b-1\}$; therefore in GSOM-BIC we pay
$p_j^s=\E[\min\{\tfrac12,2B_i-1\}\mid B_i\ge\tfrac12]=\tfrac38$.
Thus $p_i^b=\tfrac12>\tfrac38=p_j^s$, so GSOM-BIC is ex-post weakly budget-balanced, but not ex-post strongly budget-balanced.

\subsection{A BIC version of GBOM}

We now define a BIC variant of GBOM in which the sellers remain DSIC, while the buyers are only
BIC. The allocation rule is  the same as in GBOM; only the payments to the profit-receiving
side are changed.

\paragraph{Virtual-surplus matching for GBOM.}
Let $M^\psi(\mathbf b,\mathbf s)$ denote the unique matching selected by GBOM:
\[
M^\psi(\mathbf b,\mathbf s)\in
\argmax_{A\in\mathcal F}\sum_{(i,j)\in A}\bigl(b_i-\tilde\psi_j(s_j)\bigr),
\]
with ties broken according to the fixed global order.

Since each ironed virtual cost $\tilde\psi_j$ is monotone nondecreasing, the proofs of
Lemma~\ref{lem:mwm_monotonicity} and Lemma~\ref{lem:buyerinauction} apply verbatim after
replacing the weight $-s_j$ of seller $j$ by the transformed weight $-\tilde\psi_j(s_j)$.
Therefore, for every fixed $(\mathbf b_{-i},\mathbf s)$, if buyer $i$ is matched in
$M^\psi((b_i,\mathbf b_{-i}),\mathbf s)$ for some $b_i$, then as only $b_i$ varies she can only be
matched to a unique seller; symmetrically, for every fixed $(\mathbf b,\mathbf s_{-j})$, if seller
$j$ is matched in $M^\psi(\mathbf b,(s_j,\mathbf s_{-j}))$ for some $s_j$, then as only $s_j$
varies she can only be matched to a unique buyer.

Accordingly, for every matched edge $(i,j)\in M^\psi(\mathbf b,\mathbf s)$, define the threshold
payments in the original GBOM by
\begin{align}
\bar b_{ij}^\psi(\mathbf b_{-i},\mathbf s)
&:=
\inf\{z:(i,j)\in M^\psi((z,\mathbf b_{-i}),\mathbf s)\},
\label{eq:gbom-buyer-threshold}
\\[2mm]
\bar s_{ij}^\psi(\mathbf b,\mathbf s_{-j})
&:=
\sup\{z:(i,j)\in M^\psi(\mathbf b,(z,\mathbf s_{-j}))\}.
\label{eq:gbom-seller-threshold}
\end{align}
Unlike in GSOM, there is in general no useful closed form for $\bar b_{ij}^\psi$ in terms of a single
seller-side slack, because the buyer-threshold also reflects competition from other buyers.

\begin{definition}[GBOM-BIC]\label{def:bgbom}
The mechanism \textbf{GBOM-BIC} is defined as follows.
\begin{enumerate}
    \item Compute the matching $M^\psi(\mathbf b,\mathbf s)$ that maximizes
    $\sum_{(i,j)\in A}(b_i-\tilde\psi_j(s_j))$.

    \item If seller $j$ is matched to buyer $i$, pay seller $j$ her threshold payment
    \[
    p_j^s(\mathbf b,\mathbf s)=\bar s_{ij}^\psi(\mathbf b,\mathbf s_{-j}).
    \]
    If seller $j$ is unmatched, set $p_j^s(\mathbf b,\mathbf s)=0$.

    \item If buyer $i$ is matched to seller $j$, charge buyer $i$ the conditional expectation of the
    buyer-threshold payment of the original GBOM:
    \[
    p_i^b(\mathbf b,\mathbf s)
    :=
    \E_{s_j' \sim F_j}\!\left[
        \bar b_{ij}^\psi(\mathbf b_{-i},(s_j',\mathbf s_{-j}))
        \,\middle|\,
        (i,j)\in M^\psi(\mathbf b,(s_j',\mathbf s_{-j}))
    \right].
    \]
    If buyer $i$ is unmatched, set $p_i^b(\mathbf b,\mathbf s)=0$.
\end{enumerate}
\end{definition}

\begin{lemma}\label{lem:gbom-pairwise-wbb}
Fix a buyer $i$, a seller $j$, a report $b_i$, and fixed reports
$(\mathbf b_{-i},\mathbf s_{-j})$ of all other agents.
In the original GBOM,
\begin{align*}
    &\E_{s_j'\sim F_j}\!\left[
\bar b_{ij}^\psi(\mathbf b_{-i},(s_j',\mathbf s_{-j}))
\cdot
\mathbf 1[(i,j)\in M^\psi((b_i,\mathbf b_{-i}),(s_j',\mathbf s_{-j}))]
\right]\\
&\qquad \ge
\E_{s_j'\sim F_j}\!\left[
\bar s_{ij}^\psi((b_i,\mathbf b_{-i}),\mathbf s_{-j})
\cdot
\mathbf 1[(i,j)\in M^\psi((b_i,\mathbf b_{-i}),(s_j',\mathbf s_{-j}))]
\right].
\end{align*}
\end{lemma}

\begin{proof}
This is the pairwise inequality used in the GBOM half of the proof of ex-ante WBB in
\cite[Lemma~5.6]{BCWZ17}, after fixing $i,j,\mathbf b_{-i},\mathbf s_{-j}$ and varying only $s_j'$.
\end{proof}

\begin{theorem}\label{thm:gbom-bic}
GBOM-BIC is DSIC-S, BIC-B, IR-B, IR-S, and ex-post WBB.
Moreover, for every buyer $i$ and every report $b_i$, the interim expected payment of buyer $i$
in GBOM-BIC is  the same as in the original GBOM.
Consequently, GBOM-BIC and GBOM induce the same ex-ante total buyers' utility.
\end{theorem}

\begin{proof}
\emph{DSIC-S and IR-S.}
The allocation rule of GBOM-BIC is the GBOM allocation rule, hence it is monotone in each
seller's reported cost. Each matched seller is paid the corresponding threshold payment
$\bar s_{ij}^\psi(\mathbf b,\mathbf s_{-j})$. Therefore, by Myerson's lemma,
GBOM-BIC is DSIC-S and IR-S.

\noindent
\emph{BIC-B and equality of interim expected payments with GBOM.}
Fix a buyer $i$ and a report $b_i$.
For every seller $j$ and every fixed reports $(\mathbf b_{-i},\mathbf s_{-j})$ of all other agents,
by definition of GBOM-BIC,
\begin{align*}
&\E_{s_j'\sim F_j}\!\left[
p_i^b((b_i,\mathbf b_{-i}),(s_j',\mathbf s_{-j}))
\cdot
\mathbf 1[(i,j)\in M^\psi((b_i,\mathbf b_{-i}),(s_j',\mathbf s_{-j}))]
\right]
\\
&\qquad =
\E_{s_j'\sim F_j}\!\left[
\bar b_{ij}^\psi(\mathbf b_{-i},(s_j',\mathbf s_{-j}))
\cdot
\mathbf 1[(i,j)\in M^\psi((b_i,\mathbf b_{-i}),(s_j',\mathbf s_{-j}))]
\right].
\end{align*}
Now take expectation over $\mathbf s_{-j}$ and sum over all sellers $j$.
Since buyer $i$ can be matched to at most one seller in any profile, the left-hand side is 
the interim expected payment of buyer $i$ in GBOM-BIC, and the right-hand side is  the
interim expected payment of buyer $i$ in the original GBOM.

Since the allocation rule is also the same, buyer $i$'s interim utility in GBOM-BIC is  the
same as her interim utility in GBOM. The original GBOM is DSIC-B, so truthful reporting maximizes
that utility for every realization of the other agents' reports. Therefore truthful reporting also
maximizes the interim utility in GBOM-BIC, which proves BIC-B. The equality above also implies that
the ex-ante total buyers' utility is the same in the two mechanisms.

\noindent
\emph{IR-B.}
Fix a matched edge $(i,j)\in M^\psi(\mathbf b,\mathbf s)$. By definition,
\[
p_i^b(\mathbf b,\mathbf s)
=
\E_{s_j'\sim F_j}\!\left[
\bar b_{ij}^\psi(\mathbf b_{-i},(s_j',\mathbf s_{-j}))
\,\middle|\,
(i,j)\in M^\psi(\mathbf b,(s_j',\mathbf s_{-j}))
\right].
\]
For every $s_j'$ in the conditioning event, buyer $i$ wins in the original GBOM with bid $b_i$.
Hence the corresponding threshold satisfies
\[
\bar b_{ij}^\psi(\mathbf b_{-i},(s_j',\mathbf s_{-j}))\le b_i.
\]
Taking conditional expectations yields
\[
p_i^b(\mathbf b,\mathbf s)\le b_i.
\]
Thus buyer $i$'s utility is nonnegative ex post.

\noindent
\emph{Ex-post WBB.}
Fix a matched edge $(i,j)\in M^\psi(\mathbf b,\mathbf s)$.
By Lemma~\ref{lem:gbom-pairwise-wbb},
\[
\E_{s_j'\sim F_j}\!\left[
\bar b_{ij}^\psi(\mathbf b_{-i},(s_j',\mathbf s_{-j}))
\cdot
\mathbf 1[(i,j)\in M^\psi(\mathbf b,(s_j',\mathbf s_{-j}))]
\right]
\ge
\E_{s_j'\sim F_j}\!\left[
\bar s_{ij}^\psi(\mathbf b,\mathbf s_{-j})
\cdot
\mathbf 1[(i,j)\in M^\psi(\mathbf b,(s_j',\mathbf s_{-j}))]
\right].
\]
Since $\bar s_{ij}^\psi(\mathbf b,\mathbf s_{-j})=p_j^s(\mathbf b,\mathbf s)$ is constant in $s_j'$,
dividing by the probability of the conditioning event gives
\[
p_i^b(\mathbf b,\mathbf s)\ge p_j^s(\mathbf b,\mathbf s).
\]
Summing over all matched edges proves ex-post WBB:
\[
\sum_i p_i^b(\mathbf b,\mathbf s)\ge \sum_j p_j^s(\mathbf b,\mathbf s).\qedhere
\]
\end{proof}

\end{document}